\documentclass[a4paper,11pt]{article}
\usepackage{cite}
\pdfoutput=1
\usepackage{amsmath}
\usepackage{graphicx}
\usepackage{amsfonts}
\usepackage{amssymb}
\usepackage{amsmath}
\usepackage{amsthm}
\usepackage{amscd}
\usepackage{dsfont}
\usepackage{mathtools}
\usepackage{bm}
\usepackage[svgnames]{xcolor}
\usepackage[colorlinks=true,urlcolor=purple,linkcolor=violet,citecolor=magenta,bookmarks=false]{hyperref}
\usepackage{stmaryrd}
\usepackage[nomessages]{fp}
\usepackage{color}
\usepackage{hyperref}
\hypersetup{colorlinks,linkcolor={blue},citecolor={blue},urlcolor={blue}}
\usepackage{empheq}

\definecolor{orange-lpens}{RGB}{230,62,37}
\definecolor{bleu-lpens}{RGB}{41,50,117}

\numberwithin{equation}{section}
\newcommand{\bc}{\begin{center}}
\newcommand{\ec}{\end{center}}
\newcommand{\beq}{\begin{equation}}
\newcommand{\eeq}{\end{equation}}
\newcommand{\beqa}{\begin{eqnarray}}
\newcommand{\eeqa}{\end{eqnarray}}
\newcommand{\beqs}{\begin{eqnarray*}}
\newcommand{\eeqs}{\end{eqnarray*}}

\newcommand{\bi}{\begin{itemize}}
\newcommand{\ei}{\end{itemize}}

\newcommand{\bra}{\langle}
\newcommand{\ket}{\rangle}

\newcommand{\dd}{\mathrm{d}}
\def\ket#1{|#1\rangle}
\def\bra#1{\langle#1|}
\def\vev#1{\langle #1\rangle}
\def\cadremath#1{\vbox{\hrule\hbox{\vrule\kern8pt\vbox{\kern8pt
			\hbox{ {$\displaystyle #1 $ } }\kern8pt} 
			\kern8pt\vrule}\hrule}}
\newtheorem{theorem}{Theorem}[section]

\newtheorem{proposition}[theorem]{Proposition}

\newtheorem{claim}[theorem]{Claim}

\theoremstyle{definition}
\newtheorem{remark}[theorem]{Remark}
\theoremstyle{definition}

\theoremstyle{definition}
\newtheorem{definition}[theorem]{Definition}
\theoremstyle{definition}

\newcommand{\EE}{ \mathbb{E} }

\begin{document}

\begin{center}
{\LARGE \bf The Quantum Symmetric Simple Exclusion Process \\ ~ \\ in the Continuum and Free Processes}
\end{center}

\vspace{0.3cm}

\begin{center}
Denis BERNARD \footnote{Email: denis.bernard@ens.fr}
\end{center}

\noindent
{\it Laboratoire de Physique de l'Ecole Normale Sup\'erieure, CNRS, ENS \& Universit\'e PSL, Sorbonne Universit\'e, Universit\'e Paris Cit\'e, 75005 Paris, France.}
\vspace{0.3cm}

\centerline{\today}

\vspace{0.5cm}

\begin{abstract}
The quantum symmetric simple exclusion process (QSSEP) is a recent extension of the symmetric simple exclusion process, designed to model quantum coherent fluctuating effects in noisy diffusive systems. It models stochastic nearest-neighbor fermionic hopping on a lattice, possibly driven out-of-equilibrium by boundary processes. We present a direct formulation in the continuum, and establish how this formulation captures the scaling limit of the discrete version. In the continuum, QSSEP emerges as a non-commutative process, driven by free increments, conditioned on the algebra of functions on the ambiant space to encode spatial correlations. We actually develop a more general framework dealing with conditioned orbits with free increments which may find applications beyond the present context. We view this construction as a preliminary step toward formulating a quantum extension of the macroscopic fluctuation theory.
\end{abstract}
\vspace{0.3cm}

{\setlength{\parskip}{0mm} \setlength{\baselineskip}{0.175in} \tableofcontents}

\medskip
\centerline{-----------------------------}

\section{Introduction and perspective}
\label{sec:intro}

Non-equilibrium phenomena, whether classical or quantum, are ubiquitous in Nature, but their understanding remains far less profound than those at equilibrium. Over the past decade, significant conceptual progress has been made for classical non-equilibrium systems, driven by exact analysis of simple model systems, such as the symmetric simple exclusion process (SSEP) or its asymmetric counterpart (ASEP). See \cite{Derrida2007Non-equilibrium,Mallick2015The} for a review. These progresses culminated in the formulation of the macroscopic fluctuation theory (MFT), an effective theory describing transports and their fluctuations in diffusive classical systems \cite{bertini2014macroscopic}.

However, the questions whether --and how-- the macroscopic fluctuation theory might be extended to the quantum realm remains open \cite{bernard2021can}. A quantum analogue of MFT should not only describe diffusive transport and its fluctuations, but also capture quantum coherent phenomena and their fluctuations, including quantum interferences, correlations, and entanglement dynamics in out-of-equilibrium diffusive systems. Recent progress has been achieved through two complementary approaches. The first approach, based on the analysis of random quantum circuits, has led to an emerging membrane picture for entanglement production in quantum chaotic systems \cite{Zhou2019emergent,Mezei2018membrane,mezei2025entanglementmembranebrowniansyk,swann2026continuummechanicsentanglementnoisy}.  For a review, see \cite{Fisher_2023}. A parallel approach involves analyzing quantum exclusion processes, which extend classical exclusion processes into the quantum domain. See \cite{barraquand2025introduction} for a review. A paradigmatic example is the quantum symmetric simple exclusion process (QSSEP) \cite{bernard2019open}, a stochastic evolution of fermions hopping along a one-dimensional chain. While its averaged dynamics reproduces the classical SSEP, studying QSSEP has revealed remarkable features: (i) transport fluctuations in such noisy quantum systems are typically classical, with sub-leading fluctuations of quantum origin \cite{bernard2025large,albert2026universalclassical}, and (ii) off-diagonal quantum correlations persist beyond decoherence, revealing a rich structure deeply connected to free probability theory \cite{hruza2022coherent}. 

Yet, these studies rely on discrete models, followed by continuous scaling limits. Our goal is here to formulate QSSEP directly in the continuum, bypassing the discrete-to-continuum transition.

The QSSEP dynamics is fully characterized by the stochastic evolution of the $N\times N$ matrix of two-point functions, $(G_s)_{ij}:=\mathrm{Tr}(\rho_s c_i^\dag c_j)$, where $c_i^\dag, c_i$ are the fermionic annihilation-creation operators and $\rho_s$ the state (density matrix) of the system at time $s$, with $i,j=1,\cdots,N$ labeling the sites of the chain. 
Specific aspects of QSSEP ensure that its matrix of two-point functions satisfies a stochastic differential equation (SDE) --this is not true for generic quantum exclusion processes but is a special feature of QSSEP--, of the following form,
\beq \label{eq:qssep-first}
G_{s+\dd s} = e^{i\dd h_s} G_s e^{-i\dd h_s} + \mathrm{boundary\ terms},
\eeq
where the stochastic Hamiltonian increments $\dd h_s$ are some instances of structured matrix valued Brownian increments. 
The QSSEP exists in three variants, each associated with distinct algebraic structures: (i) the "periodic" case, defined on a circle, whose Hamiltonian increments is naturally associated to the loop algebra of $su(N)$ and without boundary terms, (ii) the "closed" case, defined on an interval, with Hamiltonian increments associated to $su(N)$ and without boundary terms, and finally, (iii) the "open" case, also defined on an interval, with Hamiltonian increments identical to those of the closed case but with supplementary boundary processes. In the large $N$ limit, these boundary terms fix the particle densities (denoted $n_a,n_b$) at the two boundaries, and drive the system out-of-equilibrium for $n_a\not= n_b$. See Appendix \ref{app:discrete-qssep} for a more precise definition and notation.

Analyzing the scaling limit ($N\to\infty$, with a diffusive rescaling of time and space, $t=s/N^2\in \mathbb{R}_+$ and $x=i/N\in[0,1]$) reveals that the relevant quantities are the cyclic moments of $G$ of the form $\EE[(G_s)_{i_1i_2} (G_s)_{i_2i_3} \cdots (G_s)_{i_{p+1}i_1}]$, $p\geq 0$, in which the successive indices form a cycle. They scale as $O(N^{-p})$ and depend on the index positions $x_k=i_k/N$. To capture this dependence, we introduce test diagonal matrices $\hat \Delta_k$'s and consider dressed moments,
\beq \label{eq:dressed-moments}
\EE[ (G_s\hat \Delta_1 G_s\cdots \hat\Delta_p G_s)_{ii}] \sim O(1),
\eeq
Expanding these moments linearly in $\hat \Delta_k$'s yields the cyclic moments back.
They can naturally be viewed as moments of random matrices conditioned on the sub-algebra of diagonal matrices. 

From this discussion, we identify two essential properties that QSSEP in the continuum should satisfy:
\begin{itemize}
\item[(a)] The limiting process should be defined in terms of conditioned measures. In the large size limit, the algebra of diagonal matrices becomes the algebra $L_\infty[0,1]$ of bounded functions on $[0,1]$, since to any function $\Delta\in L_\infty[0,1]$ we may associate the diagonal $N\times N$ matrix $\hat \Delta$ with entries $\hat \Delta_{ii}=\Delta(i/N)$. Thus, the scaling limit of QSSEP should be defined in terms of processes on some filtered algebra $(\mathcal{A}_t)_{t\in \mathbb{R}_+}$, conditioned on $L_\infty[0,1]$.
\item[(b)] The limiting process should be driven by free independent increments. In the large size limit, the Hamiltonian increments $\dd h_s$ become independent, identically distributed, large Gaussian matrices which, according to the general lore of random matrix theory, should be mapped to mutually free semi-circular variables $\dd X_t$, in the large $N$ limit. Thus, the limiting process should be driven by free Brownian motions, $X_t\in\mathcal{A}_t$, conditioned on $L_\infty[0,1]$.
\end{itemize}

We follow this strategy to define QSSEP in the continuum. We consider a filtration of algebras $\mathcal{A}_t$, $t\in \mathbb{R}_+$, all containing a sub-algebra $\mathcal{D}$, i.e. such $\mathcal{D}\subset\mathcal{A}_{t_1}\subset \mathcal{A}_{t_2}$ for $t_1<t_2$, equipped with a measure $\EE^\mathcal{D}$, conditioned on $\mathcal{D}$, and supporting a free Brownian motion $t\in\mathbb{R}_+\to X_t\in\mathcal{A}_t$, conditioned on $\mathcal{D}$. Such motions are made of semi-circular variables $X_t$ with identically distributed increments, $X_{t_2;t_1}:=X_{t_2}-X_{t_1}$, for $t_1<t_2$, which are mutually free over $\mathcal{D}$ if their time intervals $(t_1,t_2]$ do not overlap. They are specified by their  $\mathcal{D}$-valued variance $\EE^\mathcal{D}[X_{t_1}\Delta X_{t_2}]=(t_1\wedge t_2)\, \sigma_\epsilon(\Delta)\in\mathcal{D}$, for $\Delta\in\mathcal{D}$, where the map $\sigma_\epsilon: \mathcal{D}\to \mathcal{D}$ is a completely positive (CP) map over $\mathcal{D}$.  

Given this setup, we define a process $t\to\phi_t\in\mathcal{A}_t$, solution of the following free SDE,
\beq \label{eq:orbit-1}
\phi_{t+\dd t} = e^{i\dd X_t}\phi_te^{-i\dd X_t} + \mathrm{boundary\ conditions},
\eeq
with initial condition $\phi_0\in\mathcal{D}$. Equation \eqref{eq:orbit-1} codes for flows on adjoint orbits with free increments. 
A more precise formulation is given below. We shall actually setup a more general framework in Section \ref{sec:condition-free} and \ref{sec:unitary-orbits}, and apply it subsequently to QSSEP in Section \ref{sec:qssep-continuum}.

For QSSEP, we specify $\mathcal{D}=L_\infty[0,1]$ and the Brownian variance by choosing $\sigma_\epsilon(\Delta) = \epsilon^{-1} G_\epsilon(\Delta)$, for $\Delta\in \mathcal{D}$, with $\epsilon$ a vanishingly small regularizing parameter, and where $G_\epsilon$ is the heat kernel at time $\epsilon$ with boundary conditions depending on the variant of QSSEP we consider. Actually, to get QSSEP in the continuum, we could choose any CP-map such that $\sigma_\epsilon(\Delta)=\sigma_\epsilon(1)\Delta + \partial^2\Delta + O(\epsilon)$, and take the limit $\epsilon\to 0$. See Section \ref{sec:qssep-continuum} for a more precise description.

Conditioning on the algebra of functions on the ambiant space provides a way to restore the notion  of space in algebraic structures otherwise blind to space dependencies.

Given such process, directly defined in the continuum, we then have the equivalence:
\begin{theorem} \label{th:equivalence}
Let $G_s$ be samples of the $N\times N$ matrix of two-point functions of discrete QSSEP. Let $\hat \Delta_k$'s be series of diagonal matrices with entries $(\hat \Delta_k)_{ii}=\Delta_k(i/N)$ for $\Delta_k$ smooth functions on $[0,1]$. 
Let $\phi_t$ be samples of QSSEP in the continuum, viewed as solutions of the free stochastic differential equation \eqref{eq:orbit-1}, in the limit $\epsilon\to 0$, with boundary conditions described in Sections \ref{sec:def-qssep-more} and \ref{sec:reg-bdry}.\\
Then, the process $t\to\phi_t$ represents the scaling limit of the discrete QSSEP in the sense of the convergence of all its $L_\infty[0,1]$-valued dressed moments:
\begin{equation} \label{eq:equivalence}
\lim_{N\to\infty} \EE[(G_s\hat \Delta_1G_s\cdots \hat \Delta_p G_s)_{ii}]\Big\vert_{ i=[xN] \atop s=tN^2} = \EE^\mathcal{D}[\phi_t \Delta_1\phi_t\cdots \Delta_p \phi_t](x)
\end{equation}
for all integer $p\geq0$, and $t\in\mathbb{R}_+$ and $x\in [0,1]$. 
\end{theorem}

This theorem is proven in Section \ref{sec:equivalence} by showing that the $L_\infty[0,1]$-valued moments of $\phi_t$ satisfy the same evolution equations as the scaling limit of the discrete QSSEP moments. We assume finiteness of the moments of $\phi_t$ in the limit $\epsilon\to 0$, and a natural cyclic invariance of the local moments of $\phi_t$.

However, we first make a detour in Sections \ref{sec:condition-free} and \ref{sec:unitary-orbits} to study more general conditioned free processes. The proof of Theorem \ref{th:equivalence} will then be an application of the general setup. After having introduced the basic notions of conditioned free probability and conditioned free stochastic calculus in Section \ref{sec:condition-free}, mostly following \cite{speicher1998combinatorial,speicher2019lecture,Mingo2017Free}, we study conditioned orbits on generic $*$-algebra $\mathcal{A}$ and for generic free Brownian motions on $\mathcal{A}$, conditioned on some subalgebra $\mathcal{D}\subset\mathcal{A}$, see Section \ref{sec:unitary-orbits}. This seems to have not been studied in the free probability literature but to have nice mathematical structures and relations with the physics of open systems. 

It is clear that the above setup can be generalized to any dimensions, to any metric manifold, by replacing $\partial^2$ by the Laplacian on that manifold, or more generally, by replacing it by appropriate second order differential operators such as those generating diffusion processes. We believe that the construction of Sections \ref{sec:condition-free} and \ref{sec:unitary-orbits} has the potential for many other applications beyond QSSEP, in the physics of open systems, and potentially beyond. In parallel, the approach consisting in restoring space dependencies via a sub-algebra (here $\mathcal{D}=L_\infty[0,1]$ in QSSEP, or more generally, the algebra of bounded functions on a manifold) bears similarities with standard notions in non-commutative geometry in which points and functions over space are replaced by appropriate non-commutative algebras.

QSSEP contains classical SSEP as a sub-sector, whose scaling limit fits within macroscopic fluctuation theory \cite{bertini2014macroscopic}, which is an instance of fluctuating hydrodynamics. Our construction thus contains a sub-sector coding for a weak noise hydrodynamics, in a maybe hidden way. This raises the question: Can we explicitly construct the dynamics of the QSSEP quantum state $\rho_s$ in the continuum? This would provide a quantum extension of fluctuating hydrodynamics, at least in the weak noise limit\footnote{We hope to come back to this point in a not too far future.}.

Mixing such quantum extension of fluctuating hydrodynamics with the ongoing study of interacting quantum exclusion processes (iQSEP) \cite{BernardScopa2026} would provide a clear route toward formulating a {\it quantum mesoscopic fluctuation theory} (QMFT) \cite{bernard2021can}.

In theoretical physics, scaling limits of stochastic model systems, such as \eqref{eq:qssep-first}, may be analyzed using tools from statistical field theory. It would be interesting to bridge the free probability and field theory approaches. This might have further applications.

\section{Conditioned free processes: basics}
\label{sec:condition-free}

We start by introducing necessary notions from conditioned free probability. We essentially follow R. Speicher's review \cite{speicher1998combinatorial}, see also \cite{shlyakhtenko1996random}. Although known from the free probability literature, this background material will be needed to make Sections \ref{sec:unitary-orbits} and \ref{sec:qssep-continuum} readable. The reader expert in free probability may skip this Section and move directly to the following. Since the main part of this manuscrit are Sections \ref{sec:unitary-orbits} and \ref{sec:qssep-continuum}, other readers may also skip to the next Section for a quicker initial reading.

\subsection{Conditioned free probability}

Let $\mathcal{A}$ be a unital $*$-algebra and $\mathcal{D}\subset \mathcal{A}$ a unital $*$-subalgebra.

\medskip \noindent
$\bullet$ \underline{Conditioned measure and operator valued free cumulants.}
By analogy with the commutative case, a measure $\EE^\mathcal{D}$ over $\mathcal{A}$, conditioned on $\mathcal{D}$, is a linear map, 
\beq \label{eq:conditioned-measure}
\EE^\mathcal{D}:\, \mathcal{A}\to \mathcal{D},\ \mathrm{s.t.}\ \EE^\mathcal{D}[\Delta_1a\Delta_2]=\Delta_1 \EE^\mathcal{D}[a] \Delta_2,
\eeq
for any $a\in\mathcal{A}$ and $\Delta_1,\Delta_2\in \mathcal{D}$.
If $\mathcal{D}$ is equipped with a state $\varphi_\mathcal{D}$,  then $\EE:=\varphi_\mathcal{D}\circ \EE^\mathcal{D}$  defines an unconditioned measure on $\mathcal{A}$.

 Given $a_1,\cdots,a_{p+1} \in \mathcal{A}$, their $\mathcal{D}$-valued moments are defined as the conditioned expectation values $C_{p+1}^{(a_1,\cdots,a_{p+1})}[\Delta_1,\cdots,\Delta_p]:= \EE^\mathcal{D}[a_1\Delta_1a_2\cdots \Delta_{p}a_{p+1}]\in\mathcal{D}$, with all $\Delta_k\in\mathcal{D}$. Similarly as in the unconditioned case, $\mathcal{D}$-valued free cumulants $\kappa_\pi\in \mathcal{D}$ are recursively defined by the moment-cumulant relation \cite{speicher1998combinatorial}:
 \beq \label{eq:def-cumulants}
 \EE^\mathcal{D}[a_1\Delta_1a_2\cdots \Delta_{p}a_{p+1}] =: \sum_{\pi\in NC_{p+1}} \kappa_\pi^{(a_1,\cdots,a_{p+1})}[\Delta_1,\cdots,\Delta_{p}],
 \eeq
 where the sum is over the non-crossing partitions  $\pi$ of  $p+1$ ordered objects. Of course $C_{p+1}$ and $\kappa_\pi$ depend on the chosen collection of random variables $a_1,a_{2},\cdots$. We shall suppress this explicit dependence when there is no ambiguity.  The $\kappa_\pi$'s inherit in \eqref{eq:def-cumulants} the nested structure of non-crossing partitions, and  hence are all determined from the elementary cumulants $\kappa_{p+1}$ associated to the single block partitions of $p+1$ elements, see \cite{speicher2019lecture,Nica2006Lecture,biane2003freeprobabilitycombinatorics,speicher1998combinatorial}. For instance:
\begin{align*} 
C_1&= \kappa_1,\\
C_2[\Delta]&= \kappa_2[\Delta]+ \kappa_1\Delta\kappa_1,\\
C_3[\Delta_1,\Delta_2]&= \kappa_3[\Delta_1,\Delta_2] + \kappa_2[\Delta_1\kappa_1\Delta_2]+ \kappa_1\Delta_1\kappa_2[\Delta_2]+ \kappa_2[\Delta_1]\Delta_2\kappa_1+ \kappa_1\Delta_1\kappa_1\Delta_2\kappa_1,\\
C_4[\Delta_1,\Delta_2,\Delta_3]&= \kappa_4[\Delta_1,\Delta_2,\Delta_3]+ \kappa_3[\Delta_1,\Delta_2\kappa_1\Delta_3]+\kappa_3[\Delta_1\kappa_1\Delta_2,\Delta_3] + \kappa_3[\Delta_1,\Delta_2]\Delta_3 \kappa_1 \\
 &~~~ + \kappa_1\Delta_1\kappa_3[\Delta_2,\Delta_3] + \kappa_1\Delta_1\kappa_1\Delta_2\kappa_2[\Delta_3]+  \kappa_2[\Delta_1]\Delta_2\kappa_1\Delta_3\kappa_1+ \kappa_1\Delta_1\kappa_2[\Delta_2]\Delta_3\kappa_1\\
 &~~~ + \kappa_2[\Delta_1\kappa_1\Delta_2]\Delta_3\kappa_1  + \kappa_2[\Delta_1\kappa_1\Delta_2\kappa_1\Delta_3] + \kappa_1\Delta_1\kappa_2[\Delta_2\kappa_1\Delta_3]\\
 &~~~ + \kappa_2[\Delta_1\kappa_2[\Delta_2]\Delta_3]+ \kappa_2[\Delta_1]\Delta_2\kappa_2[\Delta_3]  +\kappa_1\Delta_1\kappa_1\Delta_2\kappa_1\Delta_3\kappa_1 ,
\end{align*} 
with all $\Delta_k\in\mathcal{D}$. Recall that the numbers of non-crossing partitions of order $n$ are the Catalan numbers, and  equal to $1,2,5,14,\cdots$ for $n=1,2,3,4\cdots$.
 
\medskip \noindent
$\bullet$ \underline{Freeness over $\mathcal{D}$.}
The definition of freeness in the conditioned context is similar to that in the unconditioned case \cite{speicher1998combinatorial}. Countable series of sub-algebras $\mathcal{A}_k$ are mutually free over $\mathcal{D}$ if $\EE^\mathcal{D}[a_{i_1}\cdots a_{i_n}]=0$, for any $a_{i}\in \mathcal{A}_{i}$, such $i_k\not= i_{k+1}$,  with vanishing expectation values $\EE^\mathcal{D}[a_i]=0$. For short, we shall call them $\mathcal{D}$-free variables.

Freeness is equivalent to the additivity of the $\mathcal{D}$-valued free cumulants. That is: the $\mathcal{D}$-valued  free cumulants of the sum of $\mathcal{D}$-free variables is the sum of the respective free cumulants. Alternatively, mixed free cumulants of mutually free variables vanish.

For instance, for $X_i, Y_j$ mutually free, with $\EE^\mathcal{D}[Y_j]=0$, we have: 
\beq \label{eq:mixed-formula}
\EE^\mathcal{D}[X_1Y_1X_2Y_2X_3] = \EE^\mathcal{D}[X_1\, \kappa_2^{(Y_1,Y_2)}[\kappa_1^{(X_2}]\, X_3].
\eeq
This type of formula will be useful when doing stochastic calculus for processes with $\mathcal{D}$-free increments.

\subsection{Free Fock spaces and representation of free variables}

Following \cite{speicher1998combinatorial} we present an explicit representation of $\mathcal{D}$-free variables using appropriate creation-annihilation operators. This extends Voiculescu's construction \cite{voiculescu1991limit}.
\medskip

\noindent
$\bullet$ \underline{Free Fock space over $\mathcal{D}$.}
Let $\mathcal{D}$ be a unital $*$-algebra, as above. Let $1_\mathcal{D}$ be the identity in $\mathcal{D}$ (we might simplify the notation by denoting it as $1$, when the context is clear). 
Let $M$ be a $\mathcal{D}$-bimodule under left-right multiplication, equipped with a $\mathcal{D}$-valued sesquilinear map $\langle\cdot,\cdot\rangle: M \times M \to \mathcal{D}$ such that:
\[
\langle \Delta_0 m_1\Delta_1,m_2\Delta_2\rangle = \Delta_1\langle m_1, \Delta_0^*m_2 \rangle \Delta_2 \in \mathcal{D},
\]
for $m_{1,2}\in M$ and $\Delta_{0,1,2}\in\mathcal{D}$.
We can (and will) think about it as "scalar product", $\langle m_1, m_2 \rangle = \langle \Omega_\mathcal{D}| m_1^*m_2 |\Omega_\mathcal{D}\rangle$, with respect to a reference state $|\Omega_\mathcal{D}\rangle$. Of course, the expression "scalar product" is an abuse of wordings since one has to remember that it takes value in $\mathcal{D}$, but we shall nevertheless use it. There are positivity constraints on this scalar product which we shall not detailed here, except in Remark \ref{remark:positivity} and when discussing the application to QSSEP.

The free Fock space of $M$ over $\mathcal{D}$, denote $\mathcal{F}_\mathcal{D}[M]$, is defined as:
\beq \label{eq:def-Fock}
\mathcal{F}_\mathcal{D}[M] := \mathcal{D} \oplus M \oplus (M \otimes_\mathcal{D} M) \oplus (M \otimes_\mathcal{D} M \otimes_\mathcal{D} M) \oplus \cdots .
\eeq
That is: elements of $\mathcal{F}_\mathcal{D}[M]$ are linear combinations of words of the form $ m_1\otimes m_2\otimes\cdots \otimes m_k$, with all $m_j\in M$. The space $\mathcal{F}_\mathcal{D}[M]$ is equipped with a $\mathcal{D}$-valued sesquilinear form, defined in a nested way:
\[
\langle m_1\otimes m_2\otimes\cdots \otimes m_k, \hat m_1\otimes \hat m_2\otimes\cdots \otimes \hat m_l\rangle
:= \delta_{k;l} \langle m_k, \langle m_{k-1}, \cdots \langle m_1, \hat m_1 \rangle \cdots \hat m_{k-1} \rangle  \hat m_k\rangle .
\]
The algebra $\mathcal{D}$ acts on $\mathcal{F}_\mathcal{D}[M]$ by left (right) multiplication, and we identify $\mathcal{D}$ with its left (right) action on $\mathcal{F}_\mathcal{D}[M]$. 

\medskip \noindent
$\bullet$ \underline{Creation-annihilation operators.}
Creation-annihilation operators on $\mathcal{F}_\mathcal{D}[M]$ are defined similarly as the usual one, but with modified ($q=0$) commutation relations. For any $m,m_k\in M$ and $\Delta\in \mathcal{D}$, they are defined by:
\begin{align*}
\ell^*(m) \Delta &:= m\Delta,\\
\ell^*(m) m_1\otimes m_2\otimes\cdots \otimes m_k &:= m \otimes m_1\otimes m_2\otimes\cdots \otimes m_k,\\
\ell(m) \Delta &:= 0,\\
\ell(m) m_1\otimes m_2\otimes\cdots \otimes m_k &:= \vev{m,m_1}\,  m_2\otimes\cdots \otimes m_k.
\end{align*}
Note that $\ell(m)$ is anti-linear in $m$ and $\ell^*(m)$ is linear in $m$. We have $\ell(\Delta_1m\Delta_2)=\Delta_2^*\ell(m)\Delta_1^*$ and $\ell^*(\Delta_1m\Delta_2)=\Delta_1\ell^*(m)\Delta_2$. 
They satisfy the following $(q=0)$ commutation relations, for $m,m'\in M$:
\[
\ell(m)\ell^*(m')=\vev{m,m'}.
\]
If $(e_i)_i$ is an ortho-normalized basis $M$, i.e. $\vev{e_i,e_j}= \delta_{ij} 1_\mathcal{D}$, so that $m=\sum_i e_i \vev{e_i,m}$ for any $m\in M$. Then $\sum_i \ell^*(e_i)\ell(e_i) = 1-P_\mathcal{D}$, with $P_\mathcal{D}$ the projector on $\mathcal{D}$.

The free Fock space $\mathcal{F}_\mathcal{D}[M]$ is generated by the left action of $\mathcal{D}$ and successive actions of the creation operators on $1_\mathcal{D}$: 
\begin{align*}
\Delta\cdot 1_\mathcal{D} &= \Delta\\
\ell^*(m_1)\cdots \ell^*(m_k) 1_\mathcal{D} &= m_1\otimes m_2\otimes\cdots \otimes m_k
\end{align*}
We can view $1_\mathcal{D}$ as the vacuum $\ket{\Omega_\mathcal{D}}$. The free Fock space is spanned by $\ell^*(m_1)\cdots \ell^*(m_k)\ket{\Omega_\mathcal{D}}$. The $\mathcal{\mathcal{D}}$-valued scalar product is such that $\bra{\Omega_\mathcal{D}}\Delta\ket{\Omega_\mathcal{D}}=\Delta$ and
\[
\bra{\Omega_\mathcal{D}}\ell(m_l)\cdots \ell(m_k)\, \ell^*(\hat m_1)\cdots \ell^*(\hat m_k)\ket{\Omega_\mathcal{D}}
=\delta_{k;l} \langle m_k, \langle m_{k-1}, \cdots \langle m_1, \hat m_1 \rangle \cdots \hat m_{k-1} \rangle  \hat m_k\rangle \in \mathcal{D}
\]
We set $\bra{\Omega_\mathcal{D}}\Omega_\mathcal{D}\rangle=1_\mathcal{D}$. 
The way to compute these scalar products is to use recursively the relation $\ell(m)\ell^*(m')=\vev{m,m'}$ to move the operators $\ell(m_j)$ (resp. $\ell^*(\hat m_j)$) to the right (resp. to the left). 

\medskip \noindent
$\bullet$ \underline{Free Fock spaces over a Hilbert space (generated by indeterminates).}
Let $V$ be a Hilbert space. We take $M=\mathcal{D}\otimes V \otimes \mathcal{D}$ as a $\mathcal{D}$-bimodule, with left/right action $\Delta_1(\Delta\otimes\xi\otimes \Delta')\Delta_2= \Delta_1\Delta\otimes\xi\otimes \Delta'\Delta_2$, with $\xi\in V$ and $\Delta, \Delta'\in\mathcal{D}$. We identify $1_\mathcal{D}\otimes\xi\otimes 1_\mathcal{D}\equiv \xi$ and $\Delta\otimes\xi\otimes \Delta'\equiv \Delta\xi \Delta'$. As before, we give ourself a $\mathcal{D}$-valued sesquilinear map $\vev{\xi_1,\Delta\xi_2}\in \mathcal{D}$, and define as above the creation-annihilation operators, $\ell(\xi)$ and $\ell^*(\xi)$, for $\xi\in V$.
For $(\xi^a)_a$ a basis of $V$, elements of $M$ are linear combinations of elements of the form $\Delta\xi^a\Delta'$. Similarly, elements of $\mathcal{F}_D(M)$ are linear combinations of words of the form $\Delta_0\xi^{a_1}\Delta_1\xi^{a_2}\cdots \Delta_{p} \xi^{a_p} \Delta_{p+1}$.

We define the $\mathcal{D}$-valued variance, or "metric", $g_{ab}: \mathcal{D}\to \mathcal{D}$ by:
\beq
g^{ab}(\Delta):=\vev{\xi^a,\Delta\xi^b} \in \mathcal{D}.
\eeq
We set,
\[
\mathfrak{l}^a :=\ell(\xi^a),\quad \mathfrak{l}^{a\,*}:=\ell^*(\xi^a),
\]
with the relation,
\beq \label{eq:L-commutation}
\mathfrak{l}^a \Delta \mathfrak{l}^{b\,*} = g^{ab}(\Delta) .
\eeq

The free Fock space is then generated by left actions of $\mathcal{D}$ and successive actions of the $\mathfrak{l}^{a\,*}$'s on the vacuum $\ket{\Omega_\mathcal{D}}=1_\mathcal{D}$ with the identification:
\[
\Delta_0\mathfrak{l}^{a_1\,*}\Delta_1\mathfrak{l}^{a_2\,*}\cdots \Delta_{p-1}\mathfrak{l}^{a_p\,*} \Delta_n\ket{\Omega_\mathcal{D}} \equiv \Delta_0\xi^{a_1}\Delta_1\xi^{a_2}\cdots \Delta_{p-1} \xi^{a_p} \Delta_p .
\]
Depending on the context we denote this free Fock space as $F_\mathcal{D}[V]$  or $F_\mathcal{D}[\mathfrak{l}^a,\mathfrak{l}^{a\,*}]$. It comes equipped with sesquilinear map, which can be computed using $\mathfrak{l}^a \Delta \mathfrak{l}^{b\,*} = g^{ab}(\Delta)$ and $\bra{\Omega_\mathcal{D}}\Delta\ket{\Omega_\mathcal{D}} =\Delta$. Note that $g^{ab}(1_\mathcal{D})$ may be non trivial. 

\medskip \noindent
$\bullet$ \underline{Representation of free variables.}
Let $F_\mathcal{D}[\mathfrak{l}^a,\mathfrak{l}^{a\,*}]$ be the above Fock space. We set 
\beq \label{eq:def-X}
X^a := \mathfrak{l}^a + \mathfrak{l}^{a\,*}.
\eeq
Let $\hat{\mathcal{A}}:=\vev{\!\vev{\mathcal{D};\mathfrak{l}^a,\mathfrak{l}^{a\,*}}\!}$ be the algebra generated by $\mathcal{D}$ and the creation-annihillation operators $\mathfrak{l}^a,\mathfrak{l}^{a\,*}$, with the relation $\mathfrak{l}^a \Delta \mathfrak{l}^{b\,*}= g^{ab}(\Delta)$. Let $\mathcal{A}:=\vev{\!\vev{\mathcal{D};X^a}\!}$ be the algebra generated by $\mathcal{D}$ and the operators $X^a$. Of course $\mathcal{D}\subset \mathcal{A}\subset\hat{\mathcal{A}}$.

The conditioned expectation values on $\hat{\mathcal{A}}$ are defined as the "vacuum expectation values", that is:
\begin{align*}
\EE^\mathcal{\mathcal{D}}[\Delta] &:= \bra{\Omega_\mathcal{D}}\Delta\ket{\Omega_\mathcal{D}} = \Delta  \in \mathcal{D} ,\\
\EE^\mathcal{\mathcal{D}}[W]&:= \bra{\Omega_\mathcal{D}}W\ket{\Omega_\mathcal{D}} \in \mathcal{D} ,
\end{align*}
for any word $W=\Delta_0\mathfrak{l}^{a_1\,\varepsilon_1}\Delta_1\mathfrak{l}^{a_2\,\varepsilon_2}\cdots \Delta_{p-1}\mathfrak{l}^{a_{p}\,\varepsilon_{p}} \Delta_{p}$ with $\varepsilon =\cdot,*$. It is clear that $\EE^\mathcal{D}[W]$ can be computed using the relation $\mathfrak{l}^a \Delta \mathfrak{l}^{b\,*}= g^{ab}(\Delta)$, and $\mathfrak{l}^a\ket{\Omega_\mathcal{D}}=0$ and $\bra{\Omega_\mathcal{D}}\mathfrak{l}^{a\,*}=0$. We have
\[
\EE^\mathcal{D}[X^a]=0,\quad \EE^\mathcal{D}[X^a\Delta X^b]=g^{ab}(\Delta).
\]
Moments of the $X^a$'s can be decomposed on their $\mathcal{D}$-valued free cumulants, as defined in \eqref{eq:def-cumulants}. The variables $X^a$ form a (centred) semi-circular system over $\mathcal{D}$, since only their two first free cumulants $\kappa^a_1$ and $\kappa^{ab}_2$ are non vanishing. Here, $\kappa^a_1=0$ (for centred variables) and $\kappa^{ab}_2(\Delta)=g^{ab}(\Delta)$. 

For such systems, their moments $\EE^\mathcal{D}[X^{a_1}\cdots X^{a_{p+1}}]$ decompose on non-crossing pairings only: 
\[
\EE^\mathcal{D}[X^{a_1}\Delta_1X^{a_2}\cdots \Delta_{p} X^{a_{p+1}}]=\sum_{\overset{\pi\in NC_{p+1}}{\mathrm{pairing}}} \kappa_\pi[\Delta_1,\cdots,\Delta_{p}].
\]
For non-crossing pairings, $\kappa_\pi$ is a nested product of the variances $\kappa_2^{ab}$.
To exercice ourselves, let us prove this formula: (i) replace $X^{a_1}$ by $\mathfrak{l}^{a_1}$ and $X^{a_n}$ by $\mathfrak{l}^{a_n\,*}$, since the vacuum $\ket{\omega_\mathcal{D}}$ annihilates $\mathfrak{l}$ (and its dual by $\mathfrak{l}^*$), (ii) start from $\mathfrak{l}^{a_1}$ and read toward the right up to encountering the first $\mathfrak{l}^*$ operator, so as to get $\mathfrak{l}^{a_1}\cdots \mathfrak{l}^{a_{k-1}}\Delta_{a_k} \mathfrak{l}^{a_k\,*}\cdots$, then pair $\{a_{k-1},a_k\}$ using $\mathfrak{l}^{a_{k-1}}\Delta_{a_k\,•} \mathfrak{l}^{a_k}=g^{a_{k-1}a_k}(\Delta_{a_k})$ to be left with $\mathfrak{l}^{a_1}\cdots \mathfrak{l}^{a_{k-2}}g^{a_{k-1}a_k}(\Delta_{a_k})\cdots$; (iii) start from $\mathfrak{l}^{a_{k-2}}$ and read again toward the right to hit the first $\mathfrak{l}^*$ and do the pairing as above, until $\mathfrak{l}^{a_1}$ as been paired; (iv) then either the  remaining first operator is a $\mathfrak{l}^*$ and the output is zero, or the first operator is a $\mathfrak{l}$ and start the procedure as in (i-iii). It is clear that this yields a non-crossing pairing. And we get all non-crossing pairing this way.

Finally, if $V$ is the direct sum of two sub-spaces, $V=V^{(1)}\oplus V^{(2)}$, then the two algebras $\hat{\mathcal{A}}^{(1)}:=\vev{\!\vev{\mathcal{D};\mathfrak{l}^{(1)},\mathfrak{l}^{(1)\,*}}\!}$ and $\hat{\mathcal{A}}^{(2)}:=\vev{\!\vev{\mathcal{D};\mathfrak{l}^{(2)},\mathfrak{l}^{(2)\,*}}\!}$ are mutually $\mathcal{D}$-free. Indeed, if $a^{(1)}_i$ (resp. $a^{(2)}_j$) are elements of $\hat{\mathcal{A}}^{(1)}$ (resp. $\hat{\mathcal{A}}^{(2)}$) with vanishing expectation values they have zero component on $\mathcal{D}$ in the direct sum defining the full Fock space. Hence their multiple alternating products (i.e. $a^{(1)}_{i_1}a^{(2)}_{j_1}\cdots $) also have zero component on $\mathcal{D}$, and thus a vanishing expectation value.

\begin{remark} \label{remark:positivity}
Given a state $\varphi_\mathcal{D}$ on $\mathcal{D}$, we can view the free Fock space as a Hilbert space with norm $||W|\Omega_\mathcal{D}\rangle||^2= \varphi_\mathcal{D}[\langle\Omega_\mathcal{D}|W^*W|\Omega_\mathcal{D}\rangle]$. For instance for $W=\sum_a x_a \Delta_a'\mathfrak{l}^{a\,*}\Delta_a$, we have $||W|\Omega_\mathcal{D}\rangle||^2 = \sum_{a,b} x_b^*\varphi_\mathcal{D}[\Delta_b^*\, g^{ba}({\Delta_b'}^*\Delta_a')\Delta_a]x_a$.
Assuming that $\varphi_\mathcal{D}$ preserves positivity, the positivity of $||W|\Omega_\mathcal{D}\rangle||^2$ demands that $g^{ba}$ also preserves positivity. This is equivalent to demanding that $\EE^\mathcal{D}$ is a positive measure, that is $\EE^\mathcal{D}[W^*W]>0$ for any $W$. This imposes strong contraints on the possible variances. 
\end{remark}

\subsection{Operator valued free stochastic integrals}

We now extend the above construction to deal with conditioned free Brownian motions and the corresponding stochastic calculus.
\medskip

\noindent
$\bullet$ \underline{$\mathcal{D}$ valued semi-circular free Brownian motions.}
We replace the Hilbert space $V$ of previous construction by the space of functions from $\mathbb{R}_+$ to $V$, that is $L^2[\mathbb{R}_+]\otimes V$. We let $(\xi^a)$ be a basis of $V$, as above, and set $X_t^a=\mathfrak{l}_t^a+\mathfrak{l}^{a\,*}_t$, where $\mathfrak{l}^a_t=\ell(1_{[0,t]}\otimes \xi^a)$ and $\mathfrak{l}^{a\,*}_t=\ell^*(1_{[0,t]}\otimes \xi^a)$. They form a semi-circular system, with\footnote{With the usual notation $t\wedge s:=\mathrm{min}(t,s)$.}:
\[
\EE^\mathcal{D}[X_t^a]=0,\quad \EE^\mathcal{\mathcal{D}}[X^a_t\Delta X^b_s]=g_{t\wedge s}^{ab}(\Delta) \in \mathcal{D},
\]
with $g_{t\wedge s}^{ab}:\mathcal{D}\to\mathcal{D}$ their covariance matrix. The simplest case is $g_{t\wedge s}^{ab}=(t\wedge s)\,\sigma^{ab}$ (but generalizations with more involve time dependence are possible).

We let $\mathcal{A}_t=\vev{\!\vev{\mathcal{D}; X^a_s,\, s\leq t}\!}$ be the algebra generated by $\mathcal{D}$ and all $\{X^a_s$, $s\leq t\}$, up to time $t$. Clearly they form a filtration of algebras, i.e. $\mathcal{A}_s\subset \mathcal{A}_t$ for $s\leq t$. Similarly, we set $\hat{\mathcal{A}}_t=\vev{\!\vev{\mathcal{D}; \mathfrak{l}^a_s, \mathfrak{l}^{a\,*}_s,\, s\leq t}\!}$.

The increment $X^a_{t,s}:=X_t^a-X^a_s$, for $s\leq t$, are $\mathcal{D}$-free w.r.t. $A_{s}$, since their mixed $\mathcal{D}$-cumulants vanish, $\kappa_2^\mathcal{D}[X^a_{t,s},X^b_{s'}](\Delta) = g_{t\wedge s'}^{ab}(\Delta) - g_{s\wedge s'}^{ab}(\Delta) =0$, for $s'\leq s \leq t$. Alternatively, this mutual freeness follows from the last remark of previous Section since $L^2([0,t))=L^2([0,s))\oplus L^2([s,t))$. 

We let  $\dd X_t^a:= X^a_{t+\dd t}-X^a_t$ be the infinitesimal increments. By construction, $\dd X^a_t$ are free w.r.t. $\mathcal{A}_t$. We have $\EE^\mathcal{D}[\dd X^a_t\Delta \dd X^b_t]:=\dot g_t^{ab}(\Delta)\dd t $ with $g_t^{ab}:=\partial_t g_t^{ab}$.  In the simplest case with $g_{t\wedge s}^{ab}=(t\wedge s)\,\sigma^{ab}$, we have $\dot g_t^{ab}=\sigma^{ab}$.

\medskip \noindent
$\bullet$ \underline{Stochastic integrals w.r.t. $\dd X^a_t$ and It\^o rules.}
As usual, free stochastic integrals are defined by discretization and limits. For adapted processes $A_s,C_s\in \mathcal{A}_{s}$, their integral w.r.t. $\dd X^a_s$ are defined by:
\beq \label{eq:def-int-stoc}
\int_0^t\! A_s\,(\dd X^a_s)\, C_s := \lim_{M\to\infty} \sum_{k=0}^{M-1} A_{t_k} \delta X^a_{t_{k+1};t_k} C_{t_k},
\eeq
with $\delta X^a_{t_{k+1};t_k} := X^a_{t_{k+1}}-X^a_{t_k}$ and $\{t_k,\, k=1,\cdots,M\}$ a discretization of the interval $[0,t]$, finer as $M$ increases, with $0=t_0<t_1<\cdots<t_M=t$ and $\delta t_k := t_{k+1}-t_k\sim t/M$ as $M\to\infty$. We can choose $t_k=tk/M$.

The free It\^o formula follows by looking at the product of two stochastic integrals. 
Let  $F^a_t:=\int_0^t A_s(dX^a_s)B_s$ and  $G^b_t:=\int_0^t C_s(dX^b_s)D_s$ be two such integrals, with $A_s,B_s,C_s,D_s\in\mathcal{A}_s$, then:
\beq \label{eq:free-Ito-stoc}
F^a_t\,G^b_t = \int_0^t\! (A_s(\dd X^a_s)B_s)\,G^b_s + \int_0^t\! F_s\, (C_s(\dd X^a_s)D_s) + \int_0^t\! A_s\, \dot g_s^{ab}(\EE^\mathcal{D}[B_sC_s])\, D_s\, \dd s.
\eeq
Equivalently, since $\dd F^a_t=A_t(dX^a_t)B_t$ and  $\dd G^b_t=C_t(dX^b_t)D_t$, we read the free It\^o rules,
\beq \label{eq:free-Ito}
\dd (F_t^aG_t^b) = (\dd F_t^a) G_t^b + F_t^a (\dd G_t^b)  + A_t\, \dot g_t^{ab}(\EE^\mathcal{D}[B_tC_t])\, D_t\, \dd t.
\eeq
A proof can be found in \cite{biane1998stochastic} in the unconditioned case. For completeness, we present in Appendix \ref{app:free-ito-rules} a sketch of a slight extension of this proof to the conditioned case. A noticeable fact is that the free It\^o contraction involves the conditioned expectation value of the product of the intermediate variables. This will play a role when analyzing free processes. Of course it is simpler to deal with $t$-independent variances, i.e. $\dot g_t^{ab}=\sigma^{ab}$.

\section{Conditioned adjoint orbits with free increments}
\label{sec:unitary-orbits}

We now use the previous framework to describe adjoint orbits of unitary operators, solutions of conditioned stochastic differential equations with free increments. Those will later be key ingredients to construct QSSEP in the continuum.

\subsection{Unitary flows with free increments and their adjoint orbits}

$\bullet$ \underline{Setup.}
We start with a filtration of unital $*$-algebras $\mathcal{A}_t$, indexed by $t\in\mathbb{R}_+$, with $\mathcal{A}_s\subset \mathcal{A}_t$ for $s<t$, a $*$-subalgebra $\mathcal{D}\subset \mathcal{A}_0$, and a $\mathcal{D}$-valued semi-circular free Brownian motion $X_t\in \mathcal{A}_t$, with $X_t^*=X_t$ and $X_t=0$, zero mean and $\mathcal{D}$-conditioned variance,
\beq
\EE^\mathcal{D}[X_t \Delta X_s]= (t\wedge s)\, \sigma_\epsilon(\Delta) \in \mathcal{D},
\eeq 
for $\Delta\in \mathcal{D}$, and mutually $\mathcal{D}$-free increments. We assume a minimal setting in which $\mathcal{A}_t$ is generated by $X_t$ over $\mathcal{D}$, with $\mathcal{A}_0=\mathcal{D}$. In other words, the algebras $\mathcal{A}_t$ are made of linear combinations of words $\Delta_0 X_{t_1}\Delta_1X_{t_2}\cdots \Delta_p X_{t_p}\Delta_{p+1}$, with all $\Delta_k\in\mathcal{D}$ and all times $t_k\leq t$. The expectation value $\EE^\mathcal{D}$ is a conditioned expectation value over the subalgebra $\mathcal{D}$. For disjoint intervals $[t_i,t_{i+1}]$, the increments $X_{t_i,t_{i+1}}:=X_{t_{i+1}}-X_{t_i}$ generate mutually free subalgebras. We let $\dd X_t:=X_{t+\dd t}-X_t$ be the infinitesimal increments, with variance $\EE^\mathcal{D}[\dd X_t \Delta \dd X_t] = \sigma_\epsilon(\Delta)\, \dd t$.

The variance is specified by the map $\sigma_\epsilon:\mathcal{D}\to \mathcal{D}$, which is assumed to be a completely positive (CP) map on $\mathcal{D}$ to ensure appropriate positivity.  This CP-map is neither supposed to preserve the identity nor to be normalized. The index $\epsilon$ refers to a possible parametrization of this CP-map, which will later serve as a regularization. 

For $A_t,B_t,C_t\in \mathcal{A}_t$, the free It\^o rules read: 
\beq \label{eq:Ito-sigma}
A_t\,\dd X_t\, B_t\, \dd X_t\, C_t = A_t\, \sigma_\epsilon(\EE^\mathcal{D}[B_t])\, C_t\, \dd t .
\eeq
In particular, $\dd X_t  \dd X_t =\sigma_\epsilon(1)\dd t$.
\medskip

\noindent
$\bullet$ \underline{Adjoint orbits with free increments.}
We define a (so-called) evolution operator $U_t\in\mathcal{A}_t$ as solution of the following free stochastic differential equation (SDE), $\dd U_t = (i\dd X_t -\frac{\sigma_\epsilon(1)}{2}\dd t)\, U_t$, with $U_{t=0}=1$, or equivalently via,
\beq \label{eq:def-Ut}
U_t = 1 + i\int_0^t\!\! (\dd X_s)U_s - \frac{\sigma_\epsilon(1)}{2} \int_0^t\!\! \dd s\, U_s .
\eeq
Alternatively, we can view $U_t$ as defined by the free SDE, $U_{t+\dd t} U_t^{-1}= e^{i \dd X_t}$.

\begin{proposition} ~\\
-- The operator $U_t$ is unitary, $U_t^*=U_t^{-1}$.\\
-- The increment $U_tU_s^{-1}$ $(t\geq s)$ are free w.r.t. $\mathcal{A}_s$ and identically distributed. 
\end{proposition}
\begin{proof}
The proof is as in the unconditioned case \cite{biane1426833free}. By definition $U_t^{-1}$ is such that $U_t U_t^{-1}= U_t^{-1}U_t=1$. Using $0=\dd (U_tU^{-1}_t)=(\dd U_t)U^{-1}_t+U_t(\dd U^{-1}_t)+(\dd U_t)(\dd U^{-1}_t)$, or similarly $0=\dd (U^{-1}_tU_t)$, and the free It\^o rules, we get $\dd U^{-1}_t=U^{-1}_t (-i\dd X_t -\frac{\sigma_\epsilon(1)}{2}\dd t)$. In parallel, we have $\dd U_t^* = U_t^*(-i \dd X_t-\frac{\sigma_\epsilon(1)}{2}\dd t)$, since $X_t^*=X_t$. Hence $U_t^*=U_t^{-1}$ and $U_t$ is unitary.

Let $V_{t;s}:=U_tU_s^{-1}$ be the increment from $s$ to $t>s$. It satisfies $\dd V_{t;s}= (i\dd X_t -\frac{\sigma_\epsilon(1)}{2}\dd t)V_{t;s}$ with $V_{t;s}\vert_{t=s}=1$. Thus, $V_{t;s}$ is a function of the increments $X_{u;s}=X_u-X_s$, with $u>s$, and hence free w.r.t. $\mathcal{A}_s$. Since they satisfy the same SDE, with identical initial condition, $V_{t;s}$ and $U_{t-s}$ have identical distribution.
\end{proof}

\begin{definition}
Flows on adjoint orbits with free increments are defined by $t\to\phi_t:=U_t\phi_0 U^{-1}_t\in\mathcal{A}_t$, with $U_t$ as in equation \eqref{eq:def-Ut}, for some initial value $\phi_0\in \mathcal{D}$. 
\end{definition}

Under an infinitesimal time increment, $\phi_t$ transforms into $\phi_{t+\dd t}= e^{i\dd X_t} \phi_t e^{-i \dd X_t}$. It thus satisfies the free SDE:
\beq \label{eq:flow-unitary}
\dd\phi_t = i [\dd X_t, \phi_t] +\sigma_\epsilon(\EE^\mathcal{D}[\phi_t])\, \dd t  -\frac{1}{2}(\sigma_\epsilon(1)\phi_t+\phi_t\sigma_\epsilon(1))\dd t .
\eeq
This follows from the explicit expressions for $\dd U_t$ and $\dd U^{-1}_t$ and the free It\^o rules.
This formula is identical to $\dd \phi_t = i [\dd X_t, \phi_t] - \frac{1}{2} [\dd X_t,[\dd X_t,\phi_t]]$.

\subsection{Average dynamics and Lindbladian}

\noindent
$\bullet$ \underline{Average dynamics and Lindbladian.}
The mean dynamics $\phi_0\to \bar{\phi_t}:= \EE^\mathcal{D}[\phi_t] \in \mathcal{D}$ satisfies:
\beq \label{eq:mean-flow}
\dd \bar{\phi_t} = L_\epsilon(\bar{\phi_t})\, \dd t,
\eeq
with $L_\epsilon$ defined as
\beq \label{eq:Lindblad-phi}
L_\epsilon(\Delta) := \sigma_\epsilon(\Delta)  -\frac{1}{2}(\sigma_\epsilon(1)\Delta+\Delta\sigma_\epsilon(1)) .
\eeq

For any CP-map $\sigma_\epsilon$, the map $L_\epsilon$  defined in \eqref{eq:Lindblad-phi} is a Lindbladian on $\mathcal{D}$, i.e. it is the generator of a one-parameter semi-group of CP-maps on $\mathcal{D}$. See e.g \cite{lindblad1976generators} or M. Wolf's lecture notes (Theorem 7.2) \cite{wolf2012quantum}. By construction, $L_\epsilon(1)=0$.

Hence, equation \eqref{eq:mean-flow} is a well-posed evolution equation on $\mathcal{D}$.
Its interpretation is simple, in the spirit of the Stinespring's theorem, see \cite{wolf2012quantum}: The algebra generated by the variables $X_t$ can be viewed as that of an external bath interacting with a system whose algebra of observables is $\mathcal{D}$, and the equation \eqref{eq:Lindblad-phi} is the reduced dynamics once the bath degrees of freedom have been integrated out via the conditioned measure $\EE^\mathcal{D}$.

\medskip

\noindent
$\bullet$ \underline{Lindbladian's notation.} For later use, we need to introduce a few simple notations concerning Lindbladian related maps.

Any CP-map $\sigma_\epsilon$ can be represented as $\sigma_\epsilon(\Delta)=\sum_r K_r\Delta K_r^*$, for any $\Delta\in \mathcal{D}$, with $K_r$ so-called Kraus operators on $\mathcal{D}$, see  \cite{wolf2012quantum}. In particular, $\sigma_\epsilon(1)=\sum_r K_r  K_r^*$. As a consequence, we can write the Lindbladian as,
\[
L_\epsilon(\Delta)=\sum_r K_r\Delta K_r^* - \frac{1}{2}(K_r K_r^* \Delta+\Delta K_r K_r^*), 
\]
which is Lindblad's representation of the generators of one-parameter semi-group of CP-maps \cite{lindblad1976generators}. Alternatively, we may write $L_\epsilon(\Delta) = \frac{1}{2} \sum_r \left({ D_r(\Delta) K_r^* + K_r \bar D_r(\Delta) }\right)$ with $D_r$ and $\bar D_r$ derivations on $\mathcal{D}$ defined by $D_r(\Delta)=[K_r,\Delta]$ and $\bar D_r(\Delta)=[\Delta,K_r^*]$. Note that $D_r(\Delta)^*= \bar D_r(\Delta^*)$. 

As well know, a Lindbladian is not a derivation, i.e. it does not satisfy the Leibnitz rule, but we have, for $\Delta_1,\Delta_2\in\mathcal{D}$,
\beq \label{eq:L-doubleD}
L_\epsilon(\Delta_1\Delta_2) = L_\epsilon(\Delta_1)\Delta_2 + \Delta_1 L_\epsilon(\Delta_2)  + \sum_r D_r(\Delta_1)\bar D_r(\Delta_2).
\eeq

For later use, we introduce a dressed version of the Lindbladian $\mathfrak{L}_\epsilon(\Delta;d_1,d_2) $ defined by,  for $\Delta,d_1,d_2\in\mathcal{D}$,
\beq \label{eq:Ldressed}
\mathfrak{L}_\epsilon(\Delta;d_1,d_2) := \frac{1}{2}\Big({ L_\epsilon(d_1\Delta)d_2 -  L_\epsilon(d_1)\Delta d_2+ d_1 L_\epsilon(\Delta d_2)  - d_1 \Delta L_\epsilon(d_2) }\Big) .
\eeq
By construction, $\mathfrak{L}_\epsilon(\Delta;1,1)=L_\epsilon(\Delta)$ and $\mathfrak{L}_\epsilon(1;d_1,d_2)=0$. It may be expressed in terms of Kraus operators
$\mathfrak{L}_\epsilon(\Delta; d_1,d_2) = d_2L_\epsilon(\Delta)d_2 + \frac{1}{2} \sum_r \left({ D_r(d_1)\bar D_r(\Delta) d_2 + d_1 D_r(\Delta) \bar D_r(d_2)}\right)$. 

We introduce yet another operator $\mathfrak{D}_\epsilon(d_1,d_2;\Delta)$ defined by, for $\Delta,d_1,d_2\in\mathcal{D}$,
\beq \label{eq:Diff-dressed}
\mathfrak{D}_\epsilon(d_1,d_2;\Delta) := \frac{1}{2} \Big({ L_\epsilon(d_1\Delta d_2) -  d_1 L_\epsilon(\Delta d_2) - L_\epsilon(d_1\Delta) d_2 + d_1L_\epsilon(\Delta) d_2 }\Big) .
\eeq
It can be written as
$\mathfrak{D}_\epsilon(d_1,d_2;\Delta) = \frac{1}{2}\sum_r D_r(d_1)\Delta \bar D_r(d_2)$,
so that it corresponds to a dressed version of the defect to the Leibnitz rules \eqref{eq:L-doubleD}.
By construction $\mathfrak{D}_\epsilon(1,1;\Delta) =0$.

\begin{remark} \label{remark:L-simple}
These operators simplify when the Kraus operators are self-adjoint. In this case, we write $L_\epsilon(\Delta)=\sum_r [\ell_r,[\Delta,\ell_r]]$, with $\ell_r^*=\ell_r$, and $K_r=K_r^*=\sqrt{2}\ell_r$, or equivalently,
\[
L_\epsilon(\Delta)=-\sum_r \delta_r\delta_r(\Delta),
\]
 with derivatives $\delta_r(\Delta):=[\ell_r,\Delta]$ and $\bar \delta_r(\Delta):=-\delta_r(\Delta)$. In such case, it is clear that $L_\epsilon$ is the non-commutative analogue of a Laplacian.
 We have $L_\epsilon(\Delta_1\Delta_2)=\Delta_1L_\epsilon(\Delta_2)+L_\epsilon(\Delta_1)\Delta_2-2 \sum_r \delta_r(\Delta_1)\delta_r(\Delta_2)$ and 
\begin{align*}
 \mathfrak{L}_\epsilon(\Delta\,;\, d_1.d_2)& = -  \sum_r \delta_r(d_1 \delta_r(\Delta)d_2),\\
 \mathfrak{D}_\epsilon(d_1,d_2;\Delta) & = -\sum_r \delta_r(d_1)\Delta \delta_r(d_2) .
 \end{align*}
\end{remark}

\subsection{Moment dynamics}

The $\mathcal{D}$-valued moments $C_{p+1}^t$, with $p\geq 0$, are defined as,
\beq \label{eq:def-moments}
C_{p+1}^t[\Delta_1,\cdots,\Delta_{p}] :=\EE^\mathcal{D}[\phi_t\Delta_1 \phi_t\cdots \Delta_{p} \phi_t] ,
\eeq
with $\Delta_k\in \mathcal{D}$. There are $p+1$ insertions of $\phi_t$ in $C_{p+1}^t$, hence the index, and $C_1^t=\EE^\mathcal{D}[\phi_t]=:\bar \phi_t$ is the mean.  By construction, they belong to $\mathcal{D}$ and depend linearly on $\Delta_k$'s. Their time evolutions are encoded in the following:

\begin{proposition} \label{prop:eq-moments-general}
The $\mathcal{D}$-valued moments are solutions of a closed, triangular, hierarchy of evolution equations which read: 
\begin{align} \label{eq:C-general}
 \partial_t  C^{t} _{p+1}[\Delta_1,\cdots,\Delta_{p}] &= L_\epsilon(C^{t} _{p+1}[\Delta_1,\cdots,\Delta_{p}]) 
+   \sum_{j=1}^{p}  C^t_{p+1}[\Delta_1,\cdots, L_\epsilon( \Delta_j),\cdots, \Delta_{p}] \nonumber \\
&  \hskip -0.6 truecm -2 \sum_{j=1}^{p} \mathfrak{L}_\epsilon(\Delta_j\, ;\, C^t_j[\Delta_1,\cdots,\Delta_{j-1}],C^t_{p+1-j}[\Delta_{j+1},\cdots,\Delta_{p}]) \\
& \hskip -1.7 truecm + 2\sum_{j<k}^{p} C^t_{p+1-k+j}[\Delta_1,\cdots, \Delta_{j-1}, \mathfrak{D}_\epsilon(\Delta_j, \Delta_k; C^t_{k-j}[\Delta_{j+1},\cdots, \Delta_{k-1}]) , \Delta_{k+1},\cdots, \Delta_{p}] ,\nonumber   
\end{align}
 with dressed Lindbladian $\mathfrak{L}_\epsilon(\Delta\, ; \,d_1,d_2)$ and differential operator $\mathfrak{D}_\epsilon(d_1,d_2;\Delta)$ defined in \eqref{eq:Ldressed} and \eqref{eq:Diff-dressed}.
\end{proposition}

The proof is given in Appendix \ref{app:proof-eq-of-motion}, and is based on non-commutative It\^o calculus.

The closure property of this hierarchy originates in the freeness of the increments. Freeness also echoes in the non-linear terms in the second and third line of \eqref{eq:C-general}. Note that this hierarchy of equations involve only the Lindbladian $L_\epsilon$ and not directly the CP-map $\sigma_\epsilon$. This offers some flexibility in the definition of this CP-map, which we shall later use to regularize QSSEP in the continuum.

\begin{remark}
Given the $\mathcal{D}$-valued moments $C^t_{p+1}$ we can define the $\mathcal{D}$-valued free cumulants $\kappa_\pi$, as in \eqref{eq:def-cumulants}. The evolution equations \eqref{eq:C-general} then transfer into a non-linear, triangular, hierarchy of equations of motion for the elementary free cumulants $\kappa_{p+1}^t$. We shall exemplify it below in the case of QSSEP.
\end{remark}

\begin{remark}
The Lindblad evolution for the first moment is $\partial_t C_1^t=L_\epsilon(C_1^t)$, see \eqref{eq:mean-flow}.
For the second and third moments, $C_2^t[\Delta]$ and $C_3^t[\Delta_1,\Delta_2]$, we have: 
\begin{align*}
\partial_t  C_2^t[\Delta] &= L_\epsilon(C^{t} _{2}[\Delta]) +   C^t_{2}[L_\epsilon( \Delta)] - 2 \mathfrak{L}_\epsilon(\Delta;\, C^t_1,C^t_1) ,\\
\partial_t C_3^t[\Delta_1,\Delta_2] =& ~ L_\epsilon(C^{t} _{3}[\Delta_1,\Delta_2]) +   C^t_{3}[L_\epsilon( \Delta_1),\Delta_2] + C^t_{3}[\Delta_1,L_\epsilon( \Delta_2)] \\
& - 2 \mathfrak{L}_\epsilon(\Delta_1;\, C^t_1,C^t_2[\Delta_2]) - 2 \mathfrak{L}_\epsilon(\Delta_2;\, C^t_2[\Delta_1],C^t_1) + 2 C_2^t[\mathfrak{D}_\epsilon(\Delta_1,\Delta_2;C_1^t)] .
\end{align*}
And so on for higher moments. 
For $\overline{\phi_t^{p}}:=C_{p}^t[1,\cdots,1]$, we simply have $\partial_t \overline{\phi_t^{p}}= L_\epsilon(\overline{\phi_t^{p}})$.
\end{remark}

\section{QSSEP in the continuum}
\label{sec:qssep-continuum}

Although the above construction has the potential to have other applications, we now apply it to define QSSEP in the continuum.

\subsection{QSSEP: Definition in the continuum}
\label{sec:def-qssep-more}

We first present a slightly imprecise, but more flexible definition. A more precise setup, defining properly a regularization scheme ensuring positivity and the proper boundary conditions is discussed in the following Section \ref{sec:reg-bdry}. 
\medskip

\noindent
$\bullet$ \underline{\it Setup and definition.}
The QSSEP in the continuum is defined as a limit of a regularized version, which is needed to ensure positivity conditions required by the general framework of conditioned free processes. We denote by $\epsilon$ the small regularizing parameter ($\epsilon\to 0$). 

We take $\mathcal{D}=L_\infty[0,1]$, the space of bounded functions on the interval $[0,1]$. We denote by $\partial:=\partial_x$ the differential operator on (smooth enough) functions on $[0,1]$. 
We follow Section \ref{sec:unitary-orbits} to define adjoint orbits with free increments conditioned over $L_\infty[0,1]$.  We let $X_t$ be a free Brownian motion with $L_\infty[0,1]$-conditioned variance $\EE^\mathcal{D}[\dd X_t\Delta\dd X_t]=\sigma_\epsilon(\Delta)\, \dd t$, and denote the associated Lindbladian by $L_\epsilon(\Delta)=\sigma_\epsilon(\Delta)-\frac{1}{2}(\sigma_\epsilon(1)\Delta+\Delta\sigma_\epsilon(1))$.


\begin{definition} \label{def:qssep-continuum}
In the continuum, QSSEP is the limit $\epsilon\to0$ of the process $t\to\phi_t$, conditioned on $\mathcal{D}=L_\infty[0,1]$, solution of the free stochastic differential equation,
\beq \label{eq:flow-unitary-bis}
\dd\phi_t = i [\dd X_t, \phi_t] +\sigma_\epsilon(\EE^\mathcal{D}[\phi_t])\, \dd t  -\frac{1}{2}(\sigma_\epsilon(1)\phi_t+\phi_t\sigma_\epsilon(1))\dd t + \mathrm{boundary\ conditions},
\eeq
with $X_t$ a free Brownian motion, conditioned over $L_\infty[0,1]$, such that $L_\epsilon(\Delta)=\partial^2\Delta + O(\epsilon)$. Different boundary conditions specify the different variants of QSSEP.
\end{definition}

In the "periodic" and "closed" case, the flows is unitary and there is no boundary term in \eqref{eq:flow-unitary-bis}. In the "open" case, the regularization involves adding unitary-breaking boundary terms in \eqref{eq:flow-unitary-bis} at finite $\epsilon$, which enforce relaxation toward the proper boundary conditions in the limit $\epsilon\to0$. These terms are echoes of the injection-extraction boundary processes present in the discrete setting which stabilize the particle density at the boundaries. See Section \ref{sec:reg-bdry}.
\medskip

\noindent
$\bullet$ \underline{\it Boundary conditions.}
Boundary conditions are clearly needed in view of the heat equation satisfied by the one point function,  $\partial_t C_1^t=\partial^2C_1^t$, see \eqref{eq:mean-flow}. To discuss them, let us introduce the $\mathcal{D}$-valued moments $C_{p+1}^t$, as in \eqref{eq:def-moments}. They are functions on $[0,1]$,
\[
C_{p+1}^t[\Delta_1,\cdots,\Delta_p](x):=\EE^\mathbb{D}[\phi_t\Delta_1\phi_t\cdots\Delta_p\phi_t](x),\quad x\in[0,1] ,
\]
 with $\Delta_k$ functions on $[0,1]$. By linearity, we write
 \beq \label{eq:phi-moments}
 C_{p+1}^t[\Delta_1,\cdots,\Delta_p](x) =:\int_0^1\!\!(\dd x_1\cdots\dd x_p)\Delta_1(x_1)\cdots\Delta_p(x_p)\, \mathfrak{g}^t_{p+1}(x,x_1,\cdots,x_p) ,
\eeq
for some function $ \mathfrak{g}^t_{p+1}$, called local moments.  We assume that all $ \mathfrak{g}^t_{p+1}$ are invariant under cyclic permutations of their arguments (in accordance with their discrete interpretation).
 
 Using the moment-cumulant formula \eqref{eq:def-cumulants}, the moments $C_{p+1}^t$ can be written in terms of their free cumulants $\kappa^t_\pi[\Delta_1,\cdots,\Delta_p]$, indexed by non-crossing partitions. Let $\kappa^t_{p+1}:=\kappa^t_{1_{p+1}}$ be the free cumulant associated to the single block partition. By linearity in $\Delta$'s, we write
 \beq \label{eq:kappa-gp}
 \kappa^t_{p+1}[\Delta_1,\cdots,\Delta_p](x) =:\int_0^1\!(\dd x_1\cdots\dd x_p) \Delta_1(x_1)\cdots\Delta_p(x_p)\, g^t_{p+1}(x,x_1,\cdots,x_p),
 \eeq
 for some functions $g^t_p$, called cyclic loop expectation values in \cite{bernard2019open} and local free cumulants in \cite{Bernard_structured2024}. In particular $g_1^t=\kappa_1^t=C_1^t$. Since all free cumulants $\kappa^t_\pi$ can be written in terms of the $\kappa^t_q$'s, in a nested way \cite{speicher1998combinatorial}, they all can be reconstructed from the $g^t_q$. 
 
 The boundary conditions are then:
 \begin{itemize} \label{item:boundary}
 \item[(i)]  \underline{"Periodic QSSEP":} All functions are periodic on $[0,1]$. This includes all local free cumulants $g_{p}^t$, in all of its arguments, and hence all moments $C_{p}^t$. As in the discrete, the evolution is unitary and it is specified by $\phi_0$, a periodic function on $[0,1]$. 
 \item[(ii)] \underline{"Closed QSSEP":} The boundary conditions are the Neumann boundary conditions: the derivatives of the functions vanish at the boundary points $x=0,1$. This applies to all local free cumulants $g_{p}^t$ and local moments $ \mathfrak{g}_{p}^t$. In the discrete, the evolution is also unitary and associated to Neumann boundary conditions, see Appendix \ref{app:discrete-bdry}. In the continuum, the orbit is similarly characterized by a function $\phi_0$ with Neumann boundary conditions. 
 \item[(iii)] \underline{"Open QSSEP":} The systems is specified by two boundary densities $n_a$ ($n_b$) respectively at $x=0$ ($x=1)$. In terms of the local free cumulants, the boundary conditions read \cite{bernard2019open}:
 \begin{align} \label{eq:bdry-open}
 g_1^t(x)\vert_{x=0,1}&= n_{a,b}, \\
 \forall p\geq2, \ g_p^t(x_1,\cdots,x_p)&=0,\ \mathrm{if}\ \exists j\ \mathrm{s.t.}\ x_j=0,1.\nonumber
  \end{align}
 Equivalently, in terms of the local moments, the $x=0$ boundary condition reads,  
 \beq \label{eq:frak-g-bdry}
 \forall p\geq1, \ \mathfrak{g}_p^t(x_1,\cdots,x_p)\big\vert_{\exists j;\, x_j=0}= n_a^{p+1}\, \prod_{k\not= j} \delta(x_k) ,
 \eeq
 and similarly near $x=1$.
In particular,  $C_{p+1}^t[\Delta_1,\cdots,\Delta_p]\vert_{x=0,1}=n_{a,b}^{p+1}\, (\Delta_1\cdots\Delta_p)\vert_{x=0,1}$.
These boundary conditions break the unitarity. See following Section \ref{sec:reg-bdry}.
 \end{itemize}
 
 In physical terms, the two first cases correspond to equilibrium situations, while in the third case the system is driven out-of-equilibrium by the boundary conditions.
 
\begin{remark}
 The local moments $\mathfrak{g}^t_q$, in \eqref{eq:phi-moments}, should not be confused with the local free cumulants $g_q^t$, in \eqref{eq:kappa-gp}, although they are of course related  \cite{hruza2022coherent}:
\[
\mathfrak{g}_p^t(x) = \sum_{\pi\in NC_p} g_\pi^t(x)\, \delta_{\pi^*}(x),
\]
where the sum is on non-crossing partitions, with $g^t_\pi(x)=\prod_{b\in\pi}g^t_{|b|}(x)$ as usual,  $\pi^*$ the Kreweras dual of $\pi$,  and $\delta_\pi^*$ the Dirac function enforcing the points in the block of $\pi^*$ to be equal.
 \end{remark}
 
\medskip

\noindent
$\bullet$ \underline{\it Positivity and regularization.} As we are going to see below, requiring that the CP-map $\sigma_\epsilon$ be proportional to the Laplacian (up to an additive constant) ensures that the equations of motion for moments of $\phi_t$ coincide with the scaling limit of those of the discrete QSSEP. But the Laplacian is not a positive map and cannot be used to define the $\mathcal{D}$-valued variance of the Brownian motion. Some regularization is needed to ensure positivity.

Introducing a small regularizing parameter $\epsilon$, we take 
\beq \label{eq:sigma-qssep}
\sigma_\epsilon(\Delta) := \epsilon^{-1}\, G_\epsilon(\Delta) ,
\eeq
with $G_\epsilon$ the heat kernel at `time' $\epsilon$, i.e. solution of the heat equation $(\partial_\epsilon -\partial^2)G_\epsilon=0$, with $G_0=\mathrm{id}$, and some specific boundary conditions (say, periodic, Neumann, or Dirichlet, etc) which we shall discuss below. Those boundary conditions specify the domain of $G_\epsilon$. On that domain, we formally have $G_\epsilon=\exp{\epsilon \partial^2}$. 

By construction, the map $\sigma_\epsilon$ is completely positive on $L_\infty[0,1]$. The associated Lindbladian is $L_\epsilon(\Delta) := \sigma_\epsilon(\Delta)-\frac{1}{2}(\sigma_\epsilon(1)\Delta+\Delta\sigma_\epsilon(1))$. On the domain of $G_\epsilon$, we have:
\beq \label{eq:L-qssep}
L_\epsilon(\Delta) = {\epsilon}^{-1}\left({e^{\epsilon \partial^2}-1}\right)\Delta = \partial^2\Delta + O(\epsilon).
\eeq
Alternatively, we can think about the CP-map $\sigma_\epsilon$ as $\sigma_\epsilon(\Delta) = \frac{\Delta}{\epsilon} + \partial^2\Delta + O(\epsilon)$.
Of course, we could take any other choice as long as $L_\epsilon(\Delta)= \partial^2\Delta + O(\epsilon)$. For instance, we might translate the CP-map, $\sigma_\epsilon(\Delta) \leadsto c \Delta + \sigma_\epsilon(\Delta)$, for any positive function $c\in L_\infty[0,1]$, without affecting the Lindbladian (and hence the rest of the construction).

\begin{remark} \label{remark:conservation}
The flow \eqref{eq:flow-unitary-bis} without boundary terms is unitary, as it corresponds to $t\to\phi_t=U_t\phi_0 U_t^{-1}$ on an adjoint orbit of $U_t$. In such case, the global moments $\tau^t_{p}:=\int_0^1\! \dd x\, \overline{\phi_t^p}(x)$ with $\overline{\phi_t^p}:=\EE^\mathcal{D}[\phi_t^{p}]$ are conserved.  Indeed, as pointed out in Section \ref{sec:unitary-orbits}, $\partial_t \overline{\phi_t^p}= L_\epsilon(\overline{\phi_t^p})$, if the flow is defined by \eqref{eq:flow-unitary}.  In the limit $\epsilon\to 0$, this becomes $\partial_t \overline{\phi_t^p}= \partial^2 \overline{\phi_t^p}$. The above boundary conditions then ensure that the global moments are conserved in the "periodic" and "closed" case (they parametrize the orbit), but not in the "open" case. At finite $\epsilon$, for $L_\epsilon(f)=\epsilon^{-1}(G_\epsilon(f)-G_\epsilon(1)f)$, and if the flow is unitary, we have,
\[
\partial_t\tau^t_{p} = \int_0^1\!\dd x\, L_\epsilon(\overline{\phi_t^p})(x)= \epsilon^{-1}\! \int_0^1\!\dd x\dd y\, G_\epsilon(x,y)(\overline{\phi_t^p}(y)-\overline{\phi_t^p}(x))=0 , 
\]
for whatever choice of symmetric heat kernel. As a consequence, the global moments are conserved at finite $\epsilon$ in the "periodic" and "closed" cases, but not in the "open" case, since in the latter the regularization involves adding unitary-breaking processes. See Section \ref{sec:reg-bdry}.
\end{remark}

\begin{remark}
If a smooth enough function $f$ does not satisfy the proper boundary conditions, and hence is not in the domain of the heat kernel, we do not necessary have $G_\epsilon(f) = f + \epsilon \partial^2 f + O(\epsilon^2)$. However, when $\epsilon\to 0$, the difference between $\epsilon^{-1}(G_\epsilon-1)(f)$ and $\partial^2 f$ is essentially localized on a shell of width $\sqrt{\epsilon}$ close to the boundaries. Hence, as long as we are using the function $f$ to test points not strictly at the boundary this does not matter in the limit $\epsilon\to 0$. We are thus going to generically use the expansion \eqref{eq:L-qssep}, implicitly assuming that we are not testing points close to the boundaries by a distance smaller than $\sqrt{\epsilon}$. 
\end{remark}

\begin{remark}
The need to introduce a regularization opens the door for possible "anomalies" --using a wording from field theory--, in the sense that there could exist observables which naively vanish in the limit $\epsilon\to 0$ but actually don't as an echo of strong irregularities in the field or test functions. 
\end{remark}

\begin{remark}
It is clear that the above construction can be generalized to higher dimensions. For instance, we could replace the heat kernel $G_\epsilon$ by the semi-group kernel $P_\epsilon$ of a higher dimensional diffusion process and the Lindbladian $L_{\epsilon\to 0}$ by the second order differential operators --sometimes called Fokker-Planck or Dynkin operators-- generating that semi-group.
\end{remark}

\subsection{Regularization and boundary conditions}
\label{sec:reg-bdry}

The regularization prescription \eqref{eq:sigma-qssep} has to be compatible with the ("periodic", "closed", "open") boundary conditions, at least for $\epsilon\to 0$. This constrains the boundary conditions on the regularizing heat kernel $G_\epsilon$. Note that one should not confuse the regularizing kernel $G_\epsilon$ and the heat like equations for the moments as in \eqref{eq:mean-flow}.
\medskip

(i) \underline{"Periodic" case: Periodic boundary conditions.}
Since, in this case, all functions are periodic, the heat kernel $G_\epsilon$ is chosen with periodic boundary conditions (denoted $G_\epsilon^P$). To be specific:
\[
G_\epsilon(x,y)^\mathrm{Periodic} 
= \sum_{n\in \mathbb{Z}} e^{-\epsilon (2n\pi)^2}e^{i2n\pi(x-y)} .
\]
Its domain, i.e. the set of smooth periodic functions with period $1$, forms an algebra. The identity $1$ belongs to this domain (the constant function $1$ is indeed a periodic function). Thus $G^P_\epsilon(1) =1$ and $\sigma_\epsilon(1)=1/\epsilon$. In this case we could also choose this domain as the algebra over which we condition. This boundary condition is such that all expectation values, all moments, are periodic functions.
\medskip

(ii) \underline{"Closed" case: Neumann boundary conditions.}
This case is similar to the previous. To be compatible with the boundary conditions on the expectation values, the heat kernel $G_\epsilon$ is chosen with Neumann boundary conditions (denoted $G_\epsilon^N$). Namely,
\[
G_\epsilon(x,y)^\mathrm{Neumann}
= 2\sum_{n\geq 0}  e^{-\epsilon (n\pi)^2} \cos(n\pi x)\cos(n\pi y) .
\]
Its domain also forms an algebra, and the constant function $1$ belongs to it (as its derivative vanishes).  Thus $G^N_\epsilon(1) =1$ and $\sigma_\epsilon(1)=1/\epsilon$. We could also have chosen this domain as the algebra over which we condition. This choice ensures that all moments are functions with Neumann boundary conditions, after a time of order $\epsilon$ (i.e. at any time $t>0$ in the limit $\epsilon\to 0$).

Let us check this last statement on the mean $\bar \phi_t:=\EE^\mathcal{D}[\phi_t]$. At finite $\epsilon$, it satisfies $\partial_t \bar \phi_t = L_\epsilon(\bar \phi_t)$ with $L_\epsilon(f)= \epsilon^{-1} ( G^N_\epsilon(f)-G^N_\epsilon(1)f)$. By construction, the derivative $\partial G^N_\epsilon(f)$ vanishes at the boundaries. Thus $\partial L_\epsilon(f)(0)=- \epsilon^{-1} G^N_\epsilon(1)(0)\, \partial f(0)$ (and similarly at $x=1$). Since $1$ is in the domain, $G^N_\epsilon(1)(0)=1$, and $\partial L_\epsilon(f)(0)= - \epsilon^{-1} \, \partial f(0)$, so that
\[
\partial_t\, \partial \bar \phi_t(0)= L_\epsilon(\bar \phi_t)(0) =  - \epsilon^{-1} \,\partial \bar \phi_t(0) .
\]
Thus, $\partial \bar \phi_t(0)$ decreases exponentially fast to zero, in a time scale of order $\epsilon$.
\medskip

(iii) \underline{"Open" case: Dirichlet boundary conditions plus boundary terms.}
The boundary conditions \eqref{eq:bdry-open} do not guarantee the conservation of the global moments $\tau^t_{p}$. They thus reflect unitary-breaking boundary processes on top of the unitary flow \eqref{eq:flow-unitary}. See Remark \ref{remark:conservation}. If we identify $\epsilon$ with $1/N^2$ in the discrete, we should indeed expect such contributions at finite $\epsilon$. See Appendix \ref{app:discrete-bdry}. We define the process at finite $\epsilon$ by
\beq
\dd \phi_t = \dd \phi_t\big\vert_\mathrm{unit.} + \dd \phi_t\big\vert_\mathrm{bdry}, 
\eeq

The unitary part $\dd \phi_t\big\vert_\mathrm{unit.}$ is a unitary flow on the adjoint orbit, as in \eqref{eq:flow-unitary}. The simplest regularization consists in choosing the heat kernel $G_\epsilon$ with vanishing Dirichlet boundary conditions (denoted $G_\epsilon^D$). Namely,
\[
G_\epsilon(x,y)^\mathrm{Dirichlet} = 2\sum_{n> 0}  e^{-\epsilon (n\pi)^2} \sin(n\pi x)\sin(n\pi y) .
\]
Its domain is made of smooth functions vanishing at the boundaries. Although the function $1$ does not belong to this domain, we nevertheless can evaluate $G^D_\epsilon(1)$.  By construction, $G^D_\epsilon(1)$ is a bounded positive function on $[0,1]$, vanishing at the boundaries.  

The boundary term $\dd \phi_t\big\vert_\mathrm{bdry}$ is a drift, singular and localized near the boundaries, enforcing the proper boundary conditions \eqref{eq:bdry-open} or \eqref{eq:frak-g-bdry} in the limit $\epsilon\to 0$, via a relaxation mechanism. By analogy with the discrete process, see Appendix \ref{app:discrete-bdry}, we set:
\beq \label{eq:bdry-process}
\dd \phi_t\big\vert_\mathrm{bdry} := \sum_{\alpha=a,b} \frac{\nu_\alpha}{2\epsilon}\left({ \delta_\alpha^\epsilon\, (n_\alpha-\phi_t) + (n_\alpha-\phi_t)\,\delta_\alpha^\epsilon }\right) \dd t, 
\eeq
where $\alpha=a,b$ labels the boundaries at $x=0,1$ with $n_\alpha$ the boundary densities $n_{a,b}$ and $\nu_{a.b}$ some rate parameters. Here $\delta_{a,b}^\epsilon$ are functions localized near the boundaries when $\epsilon\to 0$. We choose $\delta_a^\epsilon(x):=\delta^\epsilon(x)$ and $\delta_b^\epsilon(x):=\delta^\epsilon(1-x)$, with $\delta^\epsilon$ some mollifier of the Dirac function. 

At the boundary, $ \dd \phi_t\big\vert_\mathrm{bdry}$ is the dominating process --we will show below that the contributions from $ \dd \phi_t\big\vert_\mathrm{unit.}$ are irrelevant near the boundary--, and its role is to stabilize the boundary densities, as in the discrete. For instance, looking at the average $\bar \phi_t:=\EE^\mathcal{D}[\phi_t]$, we have $\partial_t \bar \phi_t= L_\epsilon(\bar \phi_t) + \partial_t \bar \phi_t\vert_\mathrm{bdry}$, with $L_\epsilon(\bar \phi_t)$ vanishing at the boundary --see below--, and
\[
\partial_t \bar \phi_t\vert_\mathrm{bdry} = \frac{\nu_a}{\epsilon}\delta_a^\epsilon\, (n_a-\bar \phi_t) + \frac{\nu_b}{\epsilon}\delta_b^\epsilon\, (n_b-\bar \phi_t) ,
\]
which imposes the boundary condition $\bar \phi_t(0)=n_a$ and $\bar \phi_t(1)=n_b$, in the limit $\epsilon\to 0$, for whatever initial conditions. These boundary conditions hold after a time scale of order $\epsilon/\nu_{a}\delta_a^\epsilon(0)$, or $\epsilon/\nu_{b}\delta_b^\epsilon(1)$ respectively. They thus hold for any time $t>0$ in the limit $\epsilon\to 0$.

Let us now look at the higher moments $C_{p}^t$ and verify by recursion that the relaxation process \eqref{eq:bdry-process} ensures the boundary conditions \eqref{eq:frak-g-bdry} on the local moments $\mathfrak{g}_p^t$. Assume this to be true for all $q\leq p$. The contribution of the boundary process \eqref{eq:bdry-process} to $\partial_t C_{p+1}^t$ is (to simplify notation, we only write the term associated to the boundary $x=0$, and we use a slight abuse of notation for the terms with $j=0,p$)
\begin{align*}
\partial_t C_{p+1}^t\big\vert_\mathrm{bdry} &= \frac{\nu_a}{2\epsilon} \sum_{j=0}^p( \EE^\mathcal{D}[\phi_t\cdots \Delta_j\delta_a^\epsilon(n_a-\phi_t)\Delta_{j+1}\cdots \phi_t] +  \EE^\mathcal{D}[\phi_t\cdots \Delta_j(n_a-\phi_t)\delta_a^\epsilon\Delta_{j+1}\cdots \phi_t] ) \\
&= \frac{\nu_a}{\epsilon}\Big((p+1) n_a^{p+1} \delta_a^\epsilon\, (\Delta_1\cdots\Delta_p)(0) -\sum_{j=0}^p \EE^\mathcal{D}[\phi_t\Delta_1\cdots \phi_t (\Delta_j\delta_a^\epsilon) \phi_t\cdots\Delta_p \phi_t] \Big),
\end{align*}
where we used the recursion hypothesis on $C_p^t$ to simplify the last line. Assuming finiteness of the moments $C_{p+1}^t$ as $\epsilon\to 0$,  the r.h.s. contribution to $\partial_t C_{p+1}^t\big\vert_\mathrm{bdry}$ should vanish as otherwise it would diverge at small $\epsilon$.  Assuming cyclic invariance of the local moments $\mathfrak{g}_{p+1}^t$, this imposes the boundary conditions \eqref{eq:frak-g-bdry} since, in each of the terms, the function $\delta_a^\epsilon$ localizes one of the argument of the local moments at the boundary. Again, the relaxation toward these boundary values occurs in a time scale of order $\epsilon/\nu_a\delta^\epsilon_a(0)$, which vanishes as $\epsilon\to 0$.

Finally, let us now show that choosing Dirichlet boundary conditions for the heat kernel is compatible with \eqref{eq:frak-g-bdry}, in the sense that the unitary contribution to the time derivatives of the moments vanish at the boundary. Consider first the mean. In the regularized model, at $\epsilon$ finite, its time evolution is $\partial_t \bar \phi_t\vert_\mathrm{unit.} = L_\epsilon(\bar \phi_t)$ with $L_\epsilon(f)= \epsilon^{-1}( G_\epsilon^D(f)- G_\epsilon^D(1)f)$. Since, $G_\epsilon^D(f)$ vanishes at the boundaries, we have  $L_\epsilon(\bar \phi_t)\vert_{x=0,1}=0$ (provided that $\bar \phi_t$ is bounded at the boundaries, which is indeed the case), and hence $\partial_t \bar \phi_t\vert^\mathrm{unit.}_{x=0,1} =0$. This is compatible with the above analysis of the relaxation toward the time independent boundary conditions $\bar \phi_t\vert_{x=0,1}=n_{a,b}$.
Consider now the second moment $C_2^t[\Delta]$. The boundary conditions \eqref{eq:bdry-open} are  $C_2^t[\Delta]\vert_{x=0,1}= n_a^2\Delta\vert_{x=0,1}$. At finite $\epsilon$, the unitary contribution to the evolution equation is $\partial_t C_2^t[\Delta] \vert_\mathrm{unit.}= L_\epsilon(C_2^t[\Delta]) + C_2^t[L_\epsilon(\Delta)] - 2 \mathfrak{L}_\epsilon(\Delta;C^t_1,C^t_1)$, with $L_\epsilon(f)=\epsilon^{-1}(G^D_\epsilon(f)-G^D_\epsilon(1)f)$ and $\mathfrak{L}_\epsilon$ defined in \eqref{eq:L-doubleD} in terms of $L_\epsilon$. At the boundaries, we thus have $\partial_t C_2^t[\Delta]\vert^\mathrm{unit.}_{x=0,1}=C_2^t[L_\epsilon(\Delta)]\vert_{x=0,1}=0$, since $L_\epsilon(f)\vert_{x=0,1}=0$. This is compatible with the boundary condition $C_2^t[\Delta]\vert_{x=0,1}= n_a^2\Delta\vert_{x=0,1}$. Similar arguments apply for higher $p$.

\begin{remark}
One may wonder\footnote{I thank Ph. Biane for asking me this question.} what is the nature of the process defined using the heat kernel $G_\epsilon^\mathrm{Mixed}$ with mixed boundary conditions $(n_b-n_a)f(0)=n_af'(0)$ and $(n_b-n_a)f(1)=n_bf'(1)$.
For $n_b=n_a$ it reduces to the Neumann boundary conditions. The linear profile $\bar n(x)=n_a+x(n_b-n_a)$ does satisfy these conditions and corresponds to the mean steady profile, as physically expected in out-of-equilibrium situations.
The case $n_a=0, n_b=1$ is simpler. 
\end{remark}

\subsection{Correspondence with the large $N$ scaling limit of QSSEP}
\label{sec:equivalence}

We now discuss how the QSSEP in the continuum, as defined  in \ref{def:qssep-continuum}, effectively represents the scaling limit of the discrete QSSEP. That is, we prove Theorem \ref{th:equivalence}. The proof consists in showing that QSSEP moments in the continuum and in the scaling limit of the discrete version do satisfy identical evolution equations.

\begin{proposition} \label{prop:C-eqs-motion}
Let $C_{p+1}^t[\Delta_1,\cdots ,\Delta_p]:=\EE^\mathcal{D}[\phi_t \Delta_1\phi_t\cdots \Delta_p \phi_t]$ be the moments of QSSEP in the continuum, for $\Delta_k$'s smooth functions on $[0,1]$. Then, for $p\geq0$,
\begin{align} \label{eq:loop-Cn-bis}
\partial_t  C^{t} _{p+1}[\Delta_1,\cdots,\Delta_{p}] =& ~~~ \partial^2C^{t} _{p+1}[\Delta_1,\cdots,\Delta_p] 
+   \sum_{j=1}^{p}  C^t_{p+1}[\Delta_1,\cdots, \partial^2 \Delta_j,\cdots, \Delta_{p}] \\
&  \hskip -0.6 truecm -2 \sum_{j=1}^{p} \partial\left({ C^t_j[\Delta_1,\cdots,\Delta_{j-1}] (\partial \Delta_j) C^t_{p+1-j} [\Delta_{j+1},\cdots,\Delta_{p}] }\right) \nonumber \\
&  \hskip -1.7 truecm + 2 \sum_{j<k}^{p} C^t_{p+1-k+j}[\Delta_1,\cdots, \Delta_{j-1}, (\partial \Delta_j) C^t_{k-j}[\Delta_{j+1},\cdots, d_{k-1}] (\partial \Delta_k) , \Delta_{k+1},\cdots, \Delta_{p}] \nonumber 
\end{align}
\end{proposition}
\begin{proof}
The proof follows from the general theory and Proposition \ref{prop:eq-moments-general}. In the case of QSSEP, we have $L_\epsilon(\Delta) = \partial^2 \Delta + O(\epsilon)$ and, by substituting $\delta_r$ by $i\partial$ in the Remark \ref{remark:L-simple}, we get
\begin{align*}
\mathfrak{L}_\epsilon(\Delta; d_1,d_2) &= \partial(d_1(\partial \Delta) d_2) + O(\epsilon),\\
\mathfrak{D}_\epsilon(d_1,d_2;\Delta) &= (\partial d_1) \Delta (\partial d_2) + O(\epsilon).
\end{align*}
Equation \eqref{eq:loop-Cn-bis} then follows from \eqref{eq:C-general} in Proposition \ref{prop:eq-moments-general}.
\end{proof}

Explicit expressions for first few values of $p=0,1,2$ are written in Appendix \ref{app:proof-equivalence}.

To establish the equivalence, we have to prove the following:

\begin{proposition} \label{prop:C-motion-discrete-qssep}
Let $\hat \Delta_k$ be diagonal $N\times N$ matrices, with entries $(\hat \Delta_k)_{ii}=\Delta_k(i/N)$ for some smooth functions $\Delta_k$ on $[0,1]$. Let $G_s$ be the discrete QSSEP $N\times N$ matrix. Then, for $x\in[0,1]$, and $t$ the diffusively rescaled time, $t=sN^{-2}$, the scaling limit of QSSEP moments,
\[
\hat C^t_{p+1}[\Delta_1,\cdots,\Delta_p](x):= \lim_{N\to\infty} \EE[\langle i=[xN]| G_{[s=N^2t]}\hat \Delta_1G_{[s=N^2t]}\cdots\hat \Delta_p G_{[s=N^2t]}|i=[xN\rangle],
\] 
satisfy the equations \eqref{eq:loop-Cn-bis}.
\end{proposition}

The proof is given in Appendix \ref{app:proof-equivalence}. The equations of motion for QSSEP moments were actually obtained earlier in \cite{bernard2021dynamics,hruza2022coherent}, in a different but equivalent form. See the Remark \ref{remark:loop-eqs} below. Appendix \ref{app:proof-equivalence} provides an alternative, more direct, proof of this statement.

Combining Propositions \ref{prop:C-eqs-motion} and \ref{prop:C-motion-discrete-qssep} proves the equivalence Theorem \ref{th:equivalence}.

\begin{remark} \label{remark:loop-eqs}
Equations \eqref{eq:loop-Cn-bis} can be written in terms of the local moments, as defined in \eqref{eq:phi-moments}.
After integration by parts, they translate into:
\beq \label{eq:diff-gp-local}
\Big({\partial_t-\sum_j\partial_j^2}\Big)\mathfrak{g}^t_p(x_1,\cdots,x_p)= 2\sum_{i<j}\partial_i\partial_j\Big({\delta(x_i-x_j)\,\mathfrak{g}^t_{j-1}(x_i,\cdots,x_{j-1})\mathfrak{g}^t_{p+i-j}(x_j,\cdots,x_{i-1}) }\Big),
\eeq
where we used an implicit cyclic notation in the argument of the $\mathfrak{g}^t_q$'s. We there recognize the equations of motion for the local moments written in \cite{bernard2021dynamics,hruza2022coherent}. 
\end{remark}

\section{Applications}

Let us present a few simple applications of this continuous formulation.  Unfortunately, although free probability yields a clear framework, it does not provide a lot of efficient computational tools, at least to our knowledge. We nevertheless shall use these applications to formulate a few open questions.

\subsection{QSSEP invariant measures}

We here aim at computing the invariance measures starting from the continuous formulation, and in particular from the evolution equations \eqref{prop:C-eqs-motion}. Of course these invariant measures are known from the discrete approach. We analyze the three versions, periodic, closed and open.

\medskip
\noindent
$\bullet$ \underline{Invariant measure in the periodic case.}
First, we show (again) that the global moments $\tau_n:=\int\!\dd x\, \overline{\phi_t^n}(x)$, with  $\overline{\phi_t^n}:=\mathbb{E}^\mathcal{D}[\phi_t^n]$, $n\geq 0$, are conserved. Equations \eqref{eq:loop-Cn-bis}, with all matrices $\Delta_k$ set to $1$, yields  $\partial_t \overline{\phi_t^n}=\partial^2 \overline{\phi_t^n}$, since $\overline{\phi_t^n}=C_n^t[1\cdots 1]$. Using periodicity, this implies that $\tau_n$ is conserved. In particular $\tau_n=\int\! \dd x\,\phi_0^n(x)$, which is interpreted as the trace of the $n$-th power of the 'matrix' $\phi_0$. Furthermore, $\overline{\phi_t^n}$ is $x$-independent in the steady measure and equal to $\tau_n$. 

By analogy with free cumulants, given the series $\tau_n$, we define numbers $\theta_p$ via the moment-cumulant relation:
\[
\tau_n =: \sum_{\pi \in NC_n} \theta_\pi ,
\] 
where the sum is over the set of non-crossing partitions, and $\theta_\pi=\prod_{b\in\pi}\theta_{|b|}$ as usual.

Let $C_{p+1}[d_1,\cdots,d_p]$ be the steady moments, with $d_k$ functions on $[0,1]$. From the previous remark, $C_1=\tau_1=\theta_1$. Let us determine the higher steady cumulants step by step.

The stationarity equation \eqref{prop:C-eqs-motion} for $C_2[d]$ reads $0 = \partial^2 C_2[d]   +  C_2[\partial^2 d] - 2 \tau_1^2\, \partial^2 d$, using $C_1=\tau_1$. We let $D_2[d] := C_2[d]-\tau_1^2\, d$. Then $\partial^2 D_2[d]+ D_2[\partial^2 d]=0$. The only periodic, linear in $d$, solution to this equation is $D_2[d]\propto \int_0^1\!\dd y\, d(y)$. Equivalently  $C_2[d](x)= c_2 [\int\!dy\, d(y)] + \tau_1^2 d(x)$, for some constant $c_2$. The latter is fixed by setting $d=1$ and using the sum rule, $C_2[1]=\tau_2$, and hence $c_2=\tau_2-\tau^2_1=\theta_2$. Finally, setting $\overline{(f)}:= \int_0^1\!dy\, f(y)$ for any $f\in L_\infty[0,1]$ --that is, the average of $f$ w.r.t. the Lebesgue measure-- we have:
\[
C_2[d]= \theta_2 \overline{(d)} + \theta_1^2 d .
\]

Using the formula for $C_1$ and $C_2$, the steady equation \eqref{prop:C-eqs-motion} for $C_3$ reads:
\begin{align*}
0 =& ~~ \partial_x^2  C_3[d_1,d_2] +  C_3[\partial^2 d_1,d_2] +  C_3[d_1,\partial^2 d_2] \\
&-2\theta_1^3[(\partial d_1) (\partial d_2) + d_1 (\partial^2 d_2) + (\partial^2 d_1) d_2)) ] \\
&-  2\theta_1\theta_2 [(\partial^2 d_1)\overline{(d_2)} + \overline{(d_1)}(\partial^2 d_2) - \overline{((\partial d_1)(\partial d_2))}] .
\end{align*}
We let $D_3[d_1,d_2]:=C_3[d_1,d_2]- \theta_1^3 d_1 d_2 -\theta_1\theta_2( \overline{(d_1)}d_2+ d_1 \overline{(d_2)} + \overline{(d_1d_2)})$. After a few integration by parts and algebraic manipulations, the steady condition becomes $\partial^2D_3[d_1,d_2]+D_3[\partial^2 d_1,d_2]+D_3[d_1,\partial^2 d_2]=0$. By the same argument as above, the only periodic, linear in $d_{1,2}$, solution is $D_3[d_1,d_2]\propto \overline{(d_1)}\,\overline{(d_2)}$. The proportionality coefficient is determined using the sum rule $C_3[1,1]=\tau_3=\theta_3 + 3\theta_1\theta_2+ \theta_1^3$ and identified with $\theta_3$. Thus
\[
C_3[d_1,d_2] =\theta_3 \overline{(d_1)}\,\overline{(d_2)}  + \theta_1\theta_2( \overline{(d_1)}d_2+ d_1 \overline{(d_2)} + \overline{(d_1d_2)}) + \theta_1^3 d_1 d_2 .
\]
We clearly see a pattern emerging, taking into account the role of non-crossing partitions.

The final result is better expressed in terms of the integrated moments, or loop expectation values:
\[
L_{p+1}[d_0,d_1,\cdots,d_p] :=\int_0^1\!\dd x\, d_0(x)C_{p+1}[d_0,d_1,\cdots,d_p] .
\]
To simplify the notation, let us denote $\overline{(d_id_j\cdots)}\,\overline{(d_k\cdots)}=[\overline{(ij\cdots)}\,\overline{(k\cdots)}]$, so that, for instance, $L_2[d_0,d_1]=\theta_2 [\overline{(0)}\,\overline{(1)}]+ \theta_1^2 [\overline{(01)}]$. Then:

\begin{claim} \label{claim:invariant}
In periodic QSSEP, the steady integrated moments are:
\beq \label{eq:inv-measure-loop}
L_{p+1}[d_0,d_1,\cdots,d_p]=\sum_{\pi\in NC_{p+1}} \theta_{\pi^*}\, [\prod_{b\in\pi}\overline{(b)}] ,
\eeq
where the sum is over the non-crossing partitions with  $\pi^*$ the Kreweras dual of $\pi$.
\end{claim}

For instance, 
\begin{align*}
L_3[d_0,d_1,d_2] &=\theta_3 [\overline{(0)}\,\overline{(1)}\,\overline{(2)}]  + \theta_1\theta_2 ([\overline{(01)}\,\overline{(2)}]+[\overline{(02)}\,\overline{(1)}]+[\overline{(0)}\,\overline{(12)}]) + \theta_1^3 [\overline{(012)}]\\
L_4[d_0,d_1,d_2,d_3] & = \theta_4 [\overline{(0)}\,\overline{(1)}\,\overline{(2)}\,\overline{(3)}] + \theta_2^2 ([\overline{(02)}\,\overline{(1)}\,\overline{(3)}]+[\overline{(0)}\,\overline{(13)}\,\overline{(2)}]) ,\\
&\hskip -1.0 truecm +\theta_1^2\theta_2([\overline{(1)}\,\overline{(023)}]+[\overline{(2)}\,\overline{(013)}]+[\overline{(012)}\,\overline{(3)}]+[\overline{(0)}\,\overline{(123)}]+[\overline{(01)}\,\overline{(23)}]+[\overline{(03)}\,\overline{(12)}]) \\
& \hskip -1.0 truecm +\theta_1\theta_3([\overline{(01)}\,\overline{(2)}\,\overline{(3)}]+[\overline{(03)}\,\overline{(1)}\,\overline{(2)}]+[\overline{(0)}\,\overline{(1)}\,\overline{(23)}]+[\overline{(0)}\,\overline{(12)}\,\overline{(3)}]) + \theta_1^4 [\overline{(0123)}] .
\end{align*}

The claim has been proved for $L_1, L_2,L_3$ above. We also did checked it for $L_4$, starting from \eqref{prop:C-eqs-motion}. 
A strategy to prove it recursively could consist to start from \eqref{prop:C-eqs-motion} and transform them into equations for the free cumulants. By recursion, one would then prove that $\kappa_{p+1}=\theta_{p+1}\, [\overline{(d_1)}\cdots \overline{(d_{p})}]$. Assuming this to be true up to $p$, the steady equation for free cumulants implies that $\partial^2\kappa_{p+1}[d_1,\cdots,d_{p}] +   \sum_{j=1}^{p}  \kappa_{p+1}[d_1,\cdots, \partial^2 d_j,\cdots, d_{p}]=0$, whose only periodic, linear in $d_k$, solution is $\kappa_{p+1}\propto [\overline{(d_1)}\cdots \overline{(d_{p})}]$. The sum rule $\tau_{p+1}=C_{p+1}[1,\cdots,1]$ then fixes the proportionality constant to be $\theta_{p+1}$. Unfortunately, we didn't find an algebraic proof of the claim \ref{claim:invariant} starting from \eqref{prop:C-eqs-motion}, although the formula \eqref{eq:inv-measure-loop} as a simple interpretation as the expectation of mutually free variables.
\medskip

\noindent
$\bullet$ \underline{Invariant measure in the closed case.}
This is very similar to the periodic case. As in the periodic case, the Neumann boundary conditions ensure that the global moments $\tau_{n}:=\int\!\dd x\,\overline{\phi_t^n}(x)$ are conserved, since \eqref{prop:C-eqs-motion} reduces to $\partial_t \overline{\phi_t^n}= \partial^2 \overline{\phi_t^n}$. The $\tau_{p+1}$'s specify the orbit.

Using arguments similar to those used in the periodic case, it is easy to verify, at least for the first few values of $p$, that the steady moments $C_{p+1}$ are identical to those of the periodic case:
\[
C_1=\theta_1,\quad C_2[d] = \theta_2 \overline{(d)} + \theta_1^2 d,\quad \mathrm{etc}.
\]
It is safe to claim that the invariant measures in the closed and periodic cases coincide.
\medskip

\noindent
$\bullet$ \underline{Invariant measure in the open case.} 
The boundary conditions in the open case are given in \eqref{eq:bdry-open}. They drive the system out-of-equilibrium. Contrary to the periodic and closed case, the integrated moments $\tau_n$ are non longer conserved. Unfortunately, we didn't find an algebraic derivation of the open QSSEP steady measure starting from its free probability formulation. Let us nevertheless briefly describe its structure to pose the problem. We stick to the case $n_a=0$, $n_b=1$, which is simpler. 

This steady measure is better describe in terms of the local free cumulants $g_{p}^t$, defined in \eqref{eq:kappa-gp}. Let $g_{p}$ their steady value. Let $\varphi$ be the Lebesgue measure on $[0,1]$ and, for $x\in[0,1]$, let $\mathbb{I}_x$ be the characteristic function of the interval $[0,x]$, viewed as random variables on $[0,1]$ with moments $\varphi( \mathbb{I}_{x_1}\cdots \mathbb{I}_{x_q})=x_1\wedge\cdots\wedge x_q$. Then, in the open QSSEP steady measure, the local free cumulants $g_{p+1}$ are the free cumulants of the $\mathbb{I}_x$'s, that is:
\beq \label{eq:qssep-steady}
g_{p}(x_1,\cdots,x_p) = \kappa_p( \mathbb{I}_{x_1},\cdots ,\mathbb{I}_{x_p}).
\eeq

This was proved by Ph. Biane \cite{biane2021combinatorics}, starting from characterizing equations for the local free cumulants derived in \cite{bernard2021solution}. The proof of \cite{biane2021combinatorics} was essentially combinatorial. Let us here present a more analytical proof following \cite{hruza2022coherent}. First, given the free local cumulants $g_p^t$, one defines auxiliary functions $\varphi_p^t$ in a way inspired by the moment-cumulant formula:
\beq \label{eq:def-varphi-gp}
\varphi_p^t(x) := \sum_{\pi\in NC_p} g_\pi^t(x).
\eeq
Second, one proves \cite{hruza2022coherent}, that they satisfy the same evolution equations as the local free cumulants, namely:
\beq \label{eq:varphi-gp-diff}
\Big({\partial_t-\sum_j\partial_j^2}\Big)\varphi^t_p(x)= 2\sum_{i<j}\delta(x_i-x_j)\,\partial_i\varphi^t_{j-1}(x_i,\cdots,x_{j-1})\, \partial_j \varphi^t_{p+i-j}(x_j,\cdots,x_{i-1}) .
\eeq
These are similar, but slightly different, to those satisfied by the local moments $\mathfrak{g}_p^t$, see \eqref{eq:diff-gp-local}.
Finally, one verifies that the functions $\varphi_p(x_1,\cdots,x_p) = x_1\wedge \cdots \wedge x_p$, are steady solutions of \eqref{eq:varphi-gp-diff}, with proper boundary conditions. The moment-cumulant relation \eqref{eq:def-varphi-gp} then proves the claim \eqref{eq:qssep-steady}. 

However, this proof certainly misses a conceptual link with free probability and does not explain the occurence of the free cumulants w.r.t. the Lebesgue measure in \eqref{eq:qssep-steady}.

\subsection{Non-equal time correlations in QSSEP}

Finally, we illustrate how to compute some correlation functions within this framework. We shall use the freeness of $U_{t+s;s}$ w.r.t. $U_s$ to get non-equal time correlation functions say $\EE[\phi_{t+s}\Delta\phi_s]$, in the steady measure (i.e. in the limit $s\to\infty$). Unfortunately, this correlation function is not the most physically relevant. We assume dealing with the open case.

We write $\EE^\mathcal{D}[\phi_{t+s}\Delta\phi_s]=\EE^\mathcal{D}[U_{t+s;s}\phi_s U_{t+s;s}^*\Delta\phi_s]$, and decomposed it on free cumulants. Mixed cumulants don't contribute, since $U_{t+s;s}$ is free w.r.t. $\phi_s$, thus:
\begin{align*}
 \EE^\mathcal{D}[(U_{t+s;s}\phi_s U_{t+s;s}^*)\Delta\phi_s] &= 
 \kappa_2^{U,U^*}[\kappa_1^\phi]\Delta \kappa_1^\phi + \kappa_1^U \kappa_1^\phi \kappa_1^{U^*} \Delta \kappa_1^\phi + \kappa_1^U \kappa_2^\phi[\kappa_1^{U^*}\Delta] \\
& = \mathbb{E}^\mathcal{D}[U_{t+s;s}\kappa_1^\phi U_{t+s;s}^*]\Delta \kappa_1^\phi+ \mathbb{E}^\mathcal{D}[U_{t+s;s}] \kappa_2^\phi[\mathbb{E}^\mathcal{D}[{U_{t+s;s}^*}]\Delta]  ,
\end{align*}
 where we suppressed time indices to simplify the notation (how to recover them is clear). 
 Letting $s\to\infty$ to reach the steady measure, we replace the cumulants at time $s$ by their steady values. So that:
 \[
 \kappa_2^\phi(x)\leadsto g_1(x)=\bar n(x),\ \kappa_2^\phi[\Delta](x)\leadsto \int_0^1\! \dd y\, g_2(x,y)\Delta(y) ,
 \]
with $\bar n(x)=n_a + x(n_b-n_a)$ and $g_2(x,y) = (n_b-n_a)^2(xy- x\wedge y)$.
We now have to compute $\EE^\mathcal{D}[U_{t+s;s}]$ and $\EE^\mathcal{D}[U_{t+s;s}\Delta U^*_{t+s;s}]$. Since $U_{t+s;s}$ is distributed as $U_t$, this amounts to compute $\EE^\mathcal{D}[U_{t}]$ and $\EE^\mathcal{D}[U_{t}\Delta U^*_{t}]$. Since $U_t$ is solution of SDE $\dd U_t = (i\dd X_t -\frac{\sigma_\epsilon(1)}{2}\dd t)\, U_t$, with $U_{t=0}=1$, we have $\dd \EE^\mathcal{D}[U_{t}]= -\frac{\sigma_\epsilon(1)}{2}\EE^\mathcal{D}[U_{t}]\dd t$, hence $\EE^\mathcal{D}[U_{t}]=e^{-t \frac{\sigma_\epsilon(1)}{2}}$, and $\EE^\mathcal{D}[U_{t}]=0$ in the limit $\epsilon \to 0$. 
To compute $\EE^\mathcal{D}[U_{t}\Delta U_t^*]=:F_2^t[\Delta]$, recall that it satisfies $\partial_t F_2^t[\Delta] = L_\epsilon(F_2^t[\Delta])$ with $L_\epsilon(\Delta) = \partial^2+O(\epsilon)$. Hence, $F_2^t[\Delta]= (e^{t\partial^2})\cdot \Delta$.
However, $\kappa_1^\phi$ is the steady solution of this heat equation, so that $F_2^t[\kappa_1^\phi]=\kappa_1^\phi$, and hence $\mathbb{E}^\mathcal{D}[U_{t+s;s}\kappa_1^\phi U_{t+s;s}^*]=\kappa_1^\phi$, for $s\to\infty$, independently of the time $t$. As a consequence
\[ 
 \EE^\mathcal{D}[\phi_{t+s}\Delta\phi_s](x)\vert_{s\to\infty} = (\bar n \Delta \bar n)(x) .
\]
Of course one may consider more complicated correlation functions.

On the physical side, these correlation functions do not give full access to the physically more relevant correlations such as $\EE[\vev{c_j^\dag(t)c_i}\vev{c_i^\dag(t)c_j}\cdots]$. The latter would be more naturally encoded in a continuous formulation of the dynamics of the QSSEP quantum state, as mentioned in the Introduction \ref{sec:intro}.

\subsection*{Acknowledgments}
I thank Philippe Biane for our many insightful discussions. This work was supported by the CNRS, the ENS, the ANR project ESQuisses under contract number ANR-20-CE47-0014-01, and by the Simons Collaboration “Probabilistic Paths to QFT”.

\appendix

\section{Discrete QSSEP: basic notation} 
\label{app:discrete-qssep}

We summarize here a minimum of information about discrete QSSEP, in a way adapted to the main text (the formulation is slightly different from the usual references). More details can be found in the review \cite{barraquand2025introduction}.
\medskip

Discrete QSSEP is a model of noisy fermionic systems, defined on a chain made of $N$ sites (it can actually be defined on any graph). Let $G_s$ be its matrix of two-point functions, $(G_s)_{ij}:=\mathrm{Tr}(\rho_s c_i^\dag c_j)$, where $\rho_s$ is the state (density matrix) of the system at time $s$ and $c_i^\dag, c_i$ the annihilation-creation operators. The indices $i,j=1,\cdots,N$ labels the sites of the chain.

Discrete QSSEP may be viewed as a model of random matrix. The time evolution for $G_s$ is $G_{s+\dd s}=e^{i\dd h_s}G_se^{-i\dd h_s}$, up to boundary terms, so that $G_s$ is solution of the following (classical) matrix-valued SDE:
\beq \label{eq:G-app}
\dd G_s = i[\dd {h}_s,G_s] -\frac{1}{2} [\dd h_s,[\dd h_s,G_s]] + \mathrm{boundary~drift~terms}.
\eeq
with Hamiltonian increment $dh_s$ of the form $\dd h_s=\sum_j( E_{j+1;j}\dd W^j_s+ E_{j;j+1}\dd\bar W^j_s)$, with $E_{i;j}=|i\rangle\langle j|$ and $W_j^s$ complex Brownian motions normalized such that $\dd W^j_s \dd \bar W^k_s=\delta^{j;k}\dd s$. There is one complex Brownian motion per edge. Here, the time $s$ is the microscopic time. When taking the scaling limit, we shall rescale space and time diffusively, $i\leadsto x=i/N$ and  $s\leadsto t:=s/N^{2}$.

There are three cases:\\
(i) \underline{"Periodic"}: There are no boundary terms \eqref{eq:G-app} and the sums are periodic, so that $\dd h_s=\sum_{j=1}^N( E_{j+1;j}\dd W^j_s+ E_{j;j+1}\dd \bar W^j_s)$ with periodic notation, e.g.  $E_{N+1;j}=E_{1;j}$.\\
(ii) \underline{"Closed"}: There are no boundary terms \eqref{eq:G-app} but also no periodicity, so that $\dd h_s=\sum_{j=1}^{N-1}( E_{j+1;j}\dd W^j_s+ E_{j;j+1}\dd \bar W^j_t)$.\\
(ii) \underline{"Open"}: The Hamiltonian increments are as in the "closed case" but there are non-Hamiltonian boundary terms:
\beq
\dd G_{ij}\vert_\mathrm{bdry}= \sum_{p=1,N} \nu_p\left({ n_p\delta_{pi}\delta_{pj} - \frac{G_{ij}}{2}(\delta_{pi}+\delta_{pj})}\right)\dd s .
\eeq
with $n_1:=n_a$ and $n_N:=n_b$ are parametrizing boundary densities, and $\nu_1=\nu_a$, $\nu_N=\nu_b$ the boundary injection/extraction rates. 

Let us write the Hamiltonian increments in a compact form as $\dd h_s=\dd {\bf W}_s+\dd {\bf \bar W}_s$, with:
\beq
\dd {\bf W}_s :=\sum_j E_{j+1;j}\dd W^j_t,\quad \dd {\bf \bar W}_s :=\sum_j E_{j;j+1}\dd\bar W^j_s .
\eeq
In matrix form, the (classical) It\^o rules read, for $A,B,C$ three $N\times N$ matrices:
\begin{subequations} \label{eq:ZIto-rules}
\begin{align}
A\, \dd{\bf W}_s\, B\, \dd {\bf \bar W}_s\, C  &= A\, \eta(SBS^\dag)\, C\, \dd s, \\
A\, \dd {\bf \bar W}_s\, B\, \dd{\bf W}_s\, C  &= A\, \eta(S^\dag BS)\, C\, \dd s ,
\end{align}
\end{subequations}
where we have defined the map $\eta$ which projects matrices to their diagonal parts via:
\beq
\eta:\ A \to \eta(A):= \mathrm{diag}( A_{jj} ) ,
\eeq
and $S$ the shift operator, with $SS^\dag= S^\dag S=1$,
\beq
S:= \sum_j E_{j+1;j} = \sum_j  |j+1\rangle\langle j| .
\eeq
We wrote the matrix It\^o contraction similarly as done in conditioned free probability, to make the analogy apparent. 

We list a few properties of $\eta$ and $S$, which make this analogy even more manifest:\\
-- For $A$ a matrix, $\eta(S^\dag AS)=S^\dag  \eta(A)S$ and $\eta(SAS^\dag )=S\eta(A)S^\dag $.\\
-- For $A,B$ two matrices, $\mathrm{tr}(A\eta(B))=\mathrm{tr}(\eta(A)B)=\mathrm{tr}(\eta(A)\eta(B))$.\\
-- For $d_1,d_2$ diagonal matrices, $\eta(d_1Ad_2)=d_1\eta(A)d_2$.\\
-- For $d$ a diagonal matrix: $[d,S]= S (\delta^+d) = (\delta^-d) S$, with $(\delta^\pm d)$ diagonal and defined as $(\delta^+d)=\mathrm{diag}(d(j+1)-d(j))$ and $(\delta^-d)=\mathrm{diag}(d(j)-d(j-1))$. 
Similarly, $[S^\dag,d]= (\delta^+d) S^\dag = S^\dag (\delta^-d)$.\\
-- For $d$ a diagonal matrix: $[S^\dag,[S,d]]=[S,[S^\dag,d]]=- (\Delta_{\mathrm{dis}}d)$, with $\Delta_{\mathrm{dis}}:=\delta^+\delta^-=\delta^-\delta^+$ the discrete Laplacian.

\section{Proof of the conditioned free It\^o rules}
\label{app:free-ito-rules}

We here present hints for a proof of the conditioned free It\^o rules \eqref{eq:free-Ito-stoc}. It is a simple application of that of \cite{biane1998stochastic} in the unconditioned case.

Recall the defining formula \eqref{eq:def-int-stoc} for stochastic integrals, where the increments $(\delta X^a_{t_{k+1};t_k})_k$ are mutually free variables. We have to prove two relations:

(i) $\sum_{k=0}^{M-1} A_k(\delta X^a_{t_{k+1};t_k})B_k(\delta X^b_{t_{k+1};t_k})C_k\to 0$, if $\EE^\mathcal{D}[B_k]=0$;

(ii) $\sum_{k=0}^{M-1} A_k(\delta X^a_{t_{k+1};t_k})\Delta_k (\delta X^b_{t_{k+1};t_k})C_k \to \int_0^t\!\dd s\, A_s\,\dot g_s^{ab}(\Delta_s)\, C_s$ ,\\
for $A_k,B_k,C_k\in \mathcal{A}_{t_k}$, $\Delta_k\in\mathcal{D}$, so that they are free w.r.t $(\delta X^a_{t_{k+1};t_k})$. 

Following R. Speicher \cite{speicher1998combinatorial},  we declare that $Z_M\to 0$ if all the free cumulants vanish in the limit, $\lim_{M\to\infty}\kappa^\mathcal{D}_n[Z_M,\cdots,Z_M]=0$, for all $n\geq 1$. Hence $Z_M\to Z$ if $\lim_{M\to\infty}\kappa^\mathcal{D}_n[Z_M-Z,\cdots,Z_M-Z]=0$, for all $n\geq 1$. 

To prove (i), we set $\Sigma_k:=A_k(\delta X^a_{t_{k+1};t_k})B_k(\delta X^b_{t_{k+1};t_k})C_k$, with $\EE^\mathcal{D}[B_k]=0$, and write $\Sigma:=\sum_{k=0}^{M-1} \Sigma_k$. We aim at proving that all the free cumulants $\kappa^\mathcal{D}_n[\Sigma,\cdots,\Sigma]\to 0$ at large $M$. We compute the first cumulants, using the formula \eqref{eq:mixed-formula} for mixed expectation values and the freeness of $(\delta X^a_{t_{k+1};t_k})_k$ w.r.t. $A_k,B_k,C_k,\Delta_k$, we get:
\begin{align*}
\EE^\mathcal{D}[ \Sigma_k]= \EE^\mathcal{D}[  A_k(\delta X^a_{t_{k+1};t_k})B_k(\delta X^b_{t_{k+1};t_k})C_k
&]= \EE^\mathcal{D}[  A_k \kappa_2^D[(\delta X^a_{t_{k+1};t_k})\EE^\mathcal{D}[B_k](\delta X^b_{t_{k+1};t_k})] C_k] \\
&=\EE^\mathcal{D}[ A_k\, \dot g_{t_k}^{ab}(\EE^\mathcal{D}[B_k])\, C_k ]\, \delta t_k .
\end{align*}
Thus $\kappa^\mathcal{D}_1[\Sigma_k]=0$ for $\EE^\mathcal{D}[B_k]=0$. We then look at $\kappa^\mathcal{D}_2[\Sigma,\Sigma](\Delta)=\sum_{k,l}\EE^\mathcal{D}[\Sigma_k\Delta\Sigma_l]$. We have: 
\[
\EE^\mathcal{D}[\Sigma_k\Delta\Sigma_l] = \EE^\mathcal{D}[A_k(\delta X^a_{t_{k+1};t_k})B_k(\delta X^b_{t_{k+1};t_k})C_k\Delta A_l(\delta X^a_{t_{l+1};t_l})B_l(\delta X^b_{t_{l+1};t_l})C_l] .
\]
The later can be computed by expanding on free cumulants via a sum on non-crossing partition. For $k<l$ (similarly for $k>l$), since $ \delta X^a_{t_{k+1};t_k}$ is free w.r.t. all other variables, the contributing non-crossing partitions have to pair $\delta X^a_{t_{k+1};t_k}$ and $\delta X^b_{t_{k+1};t_k}$, and thus producing the singlet block $\{B_k\}$. Since $\EE^\mathcal{D}[B_k]=0$, the result vanishes. Only the terms with $k=l$ might non-trivially contribute. For the terms with $k=l$, the non-zero contributions arise from pairing $\delta X^b_{t_{k+1};t_k}$ and $\delta X^a_{t_{l+1};t_l}$ in the middle together and the two extremes $\delta X^a_{t_{k+1};t_k}$ and $\delta X^b_{t_{l+1};t_l}$ together. This leads to a sum of the form $\sum_{k=0}^{M-1} \mathcal{E}_k (\delta t_k)^2$ (with $\mathcal{E}_k$ bounded, by hypothesis) which vanish in the large $M$ limit since $(\delta t_k)^2\sim 1/M^2$. Hence $\kappa^\mathcal{D}_2[\Sigma,\Sigma]=0$. Similar arguments apply for higher cumulants and $\kappa^\mathcal{D}_n[\Sigma,\cdots,\Sigma]=0$, for all $n\geq 1$.

To prove (ii), we first look at the expectation value. 
Using the formula \eqref{eq:mixed-formula} for mixed expectation values as above, we have: 
\[
\EE^\mathcal{D} [  A_k(\delta X^a_{t_{k+1};t_k})\Delta_k(\delta X^b_{t_{k+1};t_k})C_k] 
=\EE^\mathcal{D}[ A_k\, \dot g_{t_k}^{ab}(\Delta_k)\, C_k ]\, \delta t_k .
\]
By summation, the sum over $k$ converges to the integral $\int_0^t\!\dd s\, \EE^\mathcal{D}[ A_s\, \dot g_{s}^{ab}(\Delta_s)\, C_s ]$. Let now $\tilde \Sigma_k:=\hat\Sigma_k-\bar\Sigma_k$, with $\hat\Sigma_k:=A_k(\delta X^a_{t_{k+1};t_k})\Delta_k(\delta X^b_{t_{k+1};t_k})C_k$, and $\bar \Sigma_k:=A_k \dot g^{ab}(\Delta_k)C_k(\delta t_k)$, and write $\tilde\Sigma:=\sum_{k=0}^{M-1} \tilde \Sigma_k$. We aim at proving that $\kappa^\mathcal{D}_n[\tilde\Sigma,\cdots,\tilde\Sigma]=0$, for $n\geq 1$. By construction $\kappa^\mathcal{D}_1[\tilde \Sigma]=0$. For $n=2$, the second cumulants $\kappa^\mathcal{D}_2[\tilde\Sigma,\tilde\Sigma]$ is a sum over $k,l$. We have, for $k>l$,s
\[
\kappa^\mathcal{D}_2[\tilde\Sigma_k,\tilde\Sigma_l](\Delta) =
\EE^\mathcal{D}[\tilde\Sigma_k\Delta\tilde \Sigma_l] -\EE^\mathcal{D}[\tilde\Sigma_k\Delta\bar \Sigma_l] 
- \EE^\mathcal{D}[\bar\Sigma_k\Delta\tilde \Sigma_l] +\EE^\mathcal{D}[\bar\Sigma_k\Delta\bar \Sigma_l]  .
\]
Since the $(\delta X^a_{t_{k+1};t_k})$ is free w.r.t. the other variables, we have, $\EE^\mathcal{D}[\tilde\Sigma_k\Delta\tilde \Sigma_l] =\EE^\mathcal{D}[\bar\Sigma_k\Delta\tilde \Sigma_l]$ and $\EE^\mathcal{D}[\tilde\Sigma_k\Delta\bar \Sigma_l] =\EE^\mathcal{D}[\bar\Sigma_k\Delta\bar \Sigma_l]$, for $k>l$. Hence, $\kappa_2[\tilde\Sigma_k,\tilde\Sigma_l]=0$, for $k>l$. Similarly $\kappa^\mathcal{D}_2[\tilde\Sigma_k,\tilde\Sigma_l]=0$ for $k<l$. For $k=l$, the cumulants can be computed by pairing the $(\delta X^{\cdot}_{t_{k+1};t_k})$ together,  they are proportional to $(\delta t_k)^2$, and thus their sum vanish in the limit $M\to\infty$. That is $\kappa^\mathcal{D}_2[\tilde\Sigma,\tilde\Sigma]=0$. Similar arguments for higher cumulants.

Going back to the It\^o formula. $F^a_t$ and $G^b_t$ are given by limits of discrete sums: $F^a_t=\sum_{k=0}^{N-1} A_{t_k} \delta X^a_{t_{k+1};t_k} B_{t_k}$ and similarly for $G^b_t$. The product $F^a_tG^b_t$ is given by a double sum, over two indices $k,\,l$. We can split this sum according whether $k<l$, $k=l$ or $k>l$,
\[
F^a_tG^b_t = \left({ \sum_{k=0}^{M-1}\sum_{l=0}^{k-1} + \sum_{\overset{l,k=0}{l=k}}^{M-1} + \sum_{l=0}^{M-1}\sum_{k=0}^{l-1} }\right) (A_{t_k}\delta X^a_{t_{k+1};t_k} B_{t_k})(C_{t_l}\delta X^b_{t_{l+1};t_l} D_{t_l}) .
\]
The first and last sums converge to the following stochastic integrals:
\begin{align*}
 & \sum_{k=0}^{M-1} (A_{t_k}\delta X^a_{t_{k+1};t_k}B_{t_k})  (\sum_{l=0}^{k-1} C_{t_l}\delta X^a_{t_{l+1};t_l} D_{t_l}) \to \int_0^t (\dd F_s) G_s ,\\
  & \sum_{l=0}^{M-1}  (\sum_{k=0}^{l-1} A_{t_k}\delta X^a_{t_{k+1};t_k} B_{t_k}) (C_{t_l}\delta X^a_{t_{l+1};t_l} D_{t_l}) \to \int_0^t  F_s(\dd G_s)  .
\end{align*}
The middle one converges to the following integral,
\[
\sum_{k=0}^{M-1} (A_{t_k}\delta X^a_{t_{k+1};t_k} B_{t_k})(C_{t_k}\delta X^b_{t_{k+1};t_k} D_{t_k}) \to \int_0^t\!\dd s\, A_s\, \dot g_{s}^{ab}(\EE^\mathcal{D}[B_sC_s])\, D_s .
\]
This yields the It\^o formula \eqref{eq:free-Ito-stoc}. A more complete proof, filled the missing steps of the above presentation, would follow that in \cite{biane1998stochastic} in the unconditioned case.

\section{Proof of the evolution equations for moments}
\label{app:proof-eq-of-motion} 

Here is the proof of Proposition \ref{prop:eq-moments-general} and of the equations of motion \eqref{eq:C-general}.

Let $C_{p+1}^{t}[\Delta_1,\cdots,\Delta_{p}] :=\EE^\mathcal{D}[ \phi_t \Delta_1 \phi_t\cdots \Delta_{p}\phi_t ]$ be the moments.
Using $\phi_{t+\dd t} = e^{i\dd X_t}\phi_t e^{-i\dd X_t}$, we can write
\[
C_{p+1}^{t+\dd t}[\Delta_1,\cdots,\Delta_{p}] 
=\EE^\mathcal{D}[ e^{i\dd X_t} \phi_t e^{-i\dd X_t}\Delta_1 e^{i\dd X_t}\phi_t\cdots e^{-i\dd X_t}\Delta_{p}e^{i\dd X_t}\phi_t e^{-i\dd X_t}]
\]
Up to a global conjugation by $e^{\pm i\dd X_t}$, we can view this evolution equation as transforming each $\Delta_k$ into $e^{-i\dd X_t}\Delta_k e^{i\dd X_t}$. We have $e^{-i\dd X_t}\Delta e^{i\dd X_t}= \Delta - i[\dd X_t,\Delta] - \frac{1}{2}[\dd X_t,[\dd X_t,\Delta]]$.
To compute $C_{p+1}^{t+\dd t}$, we have to expand it to second order. We then have to check four relations, using the free It\^o rules \eqref{eq:Ito-sigma}:
\begin{itemize}
\item[(i)] For $[\dd X_t,[\dd X_t, \Delta]]$, with $\Delta\in\mathcal{D}$, we have:
\[
[\dd X_t,[\dd X_t, \Delta]]= - 2 L_\epsilon(\Delta)\, \dd t .
\]
We shall apply this relation for each of the $\Delta_k$'s.
\item[(ii)] For $[\dd X_t,[\dd X_t, A_t]]$, with $A_t\in\mathcal{A}_t$, we have:
\[
[\dd X_t,[\dd X_t, A_t]]= (\sigma_\epsilon(1) A_t + A_t \sigma_\epsilon(1))\dd t - 2 \sigma_\epsilon(\EE^\mathcal{D}[A_t])\, \dd t .
\]
Taking the expectation value, we get
\[
\EE^\mathcal{D}[ [\dd X_t,[\dd X_t, A_t]] ] = - 2  L_\epsilon( \EE^\mathcal{D}[A_t])\, \dd t
\] 
We shall apply this relation for $A_t= \phi_t \Delta_1 \phi_t\cdots \Delta_{p}\phi_t$.
\item[(iii)] For $A_t [\dd X_t, \Delta_1]B_t [\dd X_t, \Delta_2] C_t$, with $A_t, B_t,C_t\in\mathcal{A}_t$, and $\Delta_1,\Delta_2\in\mathcal{D}$, we have: 
\begin{align*}
A_t [\dd X_t, \Delta_1]B_t [\dd X_t, \Delta_2] C_t 
&= A_t\left({ + \sigma_\epsilon(\EE^\mathcal{D}[\Delta_1B_t]) \Delta_2 -  \sigma_\epsilon(\EE^\mathcal{D}[\Delta_1B_t \Delta_2]) }\right.\\
&~~~~~~~~   -  \left.{ \Delta_1 \sigma_\epsilon(\EE^\mathcal{D}[B_t]) \Delta_2  + \Delta_1\sigma_\epsilon(\EE^\mathcal{D}[B_t \Delta_2])  }\right)C_t\,\dd t .
\end{align*}
Recall that $\EE^\mathcal{D}[\Delta_1B_t \Delta_2]=\Delta_1 \EE^\mathcal{D}[B_t ]\Delta_2$.
We replace $\sigma_\epsilon$ for $L_\epsilon$ using $\sigma_\epsilon(\Delta) = L_\epsilon(\Delta)  +\frac{1}{2}(\sigma_\epsilon(1)\Delta+\Delta\sigma_\epsilon(1))$. All the terms with $\sigma_\epsilon(1)$ compensate, so that we are left with:
\begin{align*}
A_t [\dd X_t, \Delta_1]B_t [\dd X_t, \Delta_2] C_t 
&= A_t\left({ + L_\epsilon(\EE^\mathcal{D}[\Delta_1B_t]) \Delta_2 -  L_\epsilon(\EE^\mathcal{D}[\Delta_1B_t \Delta_2]) }\right.\\
&~~~~~~~~   -  \left.{ \Delta_1 L_\epsilon(\EE^\mathcal{D}[B_t]) \Delta_2  + \Delta_1L_\epsilon(\EE^\mathcal{D}[B_t \Delta_2])  }\right)C_t\,\dd t 
\end{align*}
We recognize (minus twice) the operator $\mathfrak{D}_\epsilon$, defined in \eqref{eq:Diff-dressed}. This becomes:
\[
A_t [\dd X_t, \Delta_1]B_t [\dd X_t, \Delta_2] C_t = - 2 A_t\, \mathfrak{D}_\epsilon(\Delta_1,\Delta_2; \EE^\mathcal{D}[B_t]) \,C_t\,\dd t .
\]
Taking the expectation value, we get:
\[
 \EE^\mathcal{D}\left[{ A_t [\dd X_t, \Delta_1]B_t [\dd X_t, \Delta_2] C_t  }\right] 
 = - 2  \EE^\mathcal{D}\left[{ A_t\, \mathfrak{D}_\epsilon(\Delta_1,\Delta_2; \EE^\mathcal{D}[B_t]) \,C_t }\right] \dd t .
\]
We shall apply this relation for $A_t=\phi_t\Delta_1\phi_t\cdots\Delta_{j-1}\phi_t$, and $B_t=\phi_t\Delta_{j+1}\phi_t\cdots\Delta_{k-1}\phi_t$, and $C_t = \phi_t\Delta_{k+1}\phi_t\cdots\Delta_{p}\phi_t$, and $\Delta_1,\Delta_2$ replaced by $\Delta_j,\Delta_k$, for each pair $j<k$.
\item[(iv)] For $[\dd X_t, A_t [\dd X_t, \Delta] C_t]$, with $A_t, C_t\in\mathcal{A}_t$ and $\Delta\in\mathcal{D}$, we have:
\begin{align*}
[\dd X_t, A_t [\dd X_t, \Delta] C_t]
= & \left({ + \sigma_\epsilon(\EE^\mathcal{D}[A_t])\Delta C_t  - \sigma_\epsilon(\EE^\mathcal{D}[A_t\Delta]) C_t }\right.\\
& ~~\left.{ - A_t \sigma_\epsilon(\EE^\mathcal{D}[\Delta C_t])  + A_t \Delta \sigma_\epsilon(\EE^\mathcal{D}[C_t]) }\right)\dd t .
\end{align*}
Replacing again $\sigma_\epsilon$ for $L_\epsilon$, we get:
\begin{align*}
[\dd X_t, A_t [\dd X_t, \Delta] C_t]
= & \left({ + L_\epsilon(\EE^\mathcal{D}[A_t])\Delta C_t  - L_\epsilon(\EE^\mathcal{D}[A_t\Delta]) C_t }\right.\\
& ~~\left.{ - A_t L_\epsilon(\EE^\mathcal{D}[\Delta C_t])  + A_t \Delta L_\epsilon(\EE^\mathcal{D}[C_t]) }\right)\dd t \\
& ~~ +\frac{1}{2} \left({ \EE^\mathcal{D}[A_t])[\sigma_\epsilon(1),\Delta] C_t - A_t [\sigma_\epsilon(1),\Delta]\EE^\mathcal{D}[C_t] }\right)\dd t .
\end{align*}
Taking the expectation value the two last terms cancel and we get:
\begin{align*}
\EE^\mathcal{D}\left[{ [\dd X_t, A_t [\dd X_t, \Delta] C_t]}\right]
= & \left({ + L_\epsilon(\EE^\mathcal{D}[A_t]) \Delta \EE^\mathcal{D}[C_t]  - L_\epsilon(\EE^\mathcal{D}[A_t\Delta]) \EE^\mathcal{D}[C_t] }\right.\\
& ~~\left.{ - \EE^\mathcal{D}[A_t] L_\epsilon(\EE^\mathcal{D}[\Delta C_t])  + \EE^\mathcal{D}[A_t] \Delta L_\epsilon(\EE^\mathcal{D}[C_t]) }\right)\dd t 
\end{align*}
We recognize (minus twice) the operator $\mathfrak{L}_\epsilon$, defined in \eqref{eq:Ldressed}. This becomes:
\[
\EE^\mathcal{D}\left[{ [\dd X_t, A_t [\dd X_t, \Delta] C_t]}\right] =
-2 \mathfrak{L}_\epsilon(\Delta\, ;\,\EE^\mathcal{D}[A_t],\EE^\mathcal{D}[C_t])\, \dd t
\]
We shall apply this relation for $A_t=\phi_t\Delta_1\phi_t\cdots\Delta_{j-1}\phi_t$, and $C_t = \phi_t\Delta_{j+1}\phi_t\cdots\Delta_{p}\phi_t$ and $\Delta=\Delta_j$, for $j=1,\cdots, p$.
\end{itemize}
Grouping all terms yields  \eqref{eq:C-general}.

\section{Proof of the equivalence with QSSEP in its scaling limit}
\label{app:proof-equivalence}

We here derive the moment evolution equations in the large $N$ scaling limit of discrete QSSEP to prove Proposition \ref{prop:C-motion-discrete-qssep}. The proof consists in using It\^o calculus and the factorization property of random matrix expectation values at large $N$.
We use the notation of Appendix \ref{app:discrete-qssep}.

Let $d_1,\cdots, d_p$, $p\geq 0$, be diagonal matrices and define the diagonal matrix $\eta(G_sd_1G_s\cdots d_{p}G_s)$. In the large $N$ scaling limit, with $s=N^2t$, $i=[Nx]$, at $x\in[0,1]$ and $t>0$ fixed, this becomes a function on $[0,1]$,
\beq
\hat Q^t_{p+1}[d_1,\cdots,d_{p}](x):= \lim_{N\to\infty} \langle j=[xN]| G_sd_1G_s\cdots d_{p}G_s|j=[xN]\rangle,\quad s=N^2t,
\eeq
which we write as $\langle x| G_sd_1G_s\cdots d_{p}G_s|x\rangle$ to simplify the notation.
 
Note that we didn't yet took the expectation value. We let $\hat C^t_{p+1}:=\EE[\hat Q^t_{p+1}]$.

The time evolution of $\hat Q^t_{p+1}$ follows from those of $G_s$, see \eqref{eq:G-app}.
We here only deal with the unitary contribution in \eqref{eq:G-app}, since the boundary terms only fix the boundary conditions in the scaling limit (see Appendix \ref{app:discrete-bdry}). Writing $G_{s+\dd s}= e^{i\dd h_s}G_s e^{-i\dd h_s}$, we have:
\[
\hat Q^{t+dt}_{p+1}[d_1,\cdots,d_p](x)= \langle x| e^{i\dd h_s} G_sd_1(\dd s)G_s\cdots d_{p}(\dd s) G_se^{-i\dd h_s}|x\rangle,
\]
where we have transfered the time evolution of $G_s$ to an action on the diagonal matrices: $d_j(\dd s):= e^{i\dd h_s}d_je^{-i\dd h_s}=d_j +i [\dd h_s,d_j] - \frac{1}{2}[\dd h_s,[\dd h_s,d_j]]$. 
\medskip

 Expanding using the (classical) It\^o rules \eqref{eq:ZIto-rules}, we get different contributions: 
\begin{itemize}
\item[$\to$] There are noisy terms, linear in $\dd h_s$, which vanish in average, thanks to the It\^o convention.
\item[$\to$] There are It\^o contractions coming from the Hamiltonian increments at the extremal positions. This leads to:
\begin{align}
&-  \frac{1}{2} \langle x| [\dd h_s,[\dd h_s,G_sd_1G_s\cdots d_{p}G]]|x\rangle 
\nonumber\\
&= -   \frac{1}{2}\langle x|  [S,[S^\dag, \eta(G_sd_1G_s\cdots d_{p}G_s )]]|x\rangle -  \frac{1}{2} \langle x| [S^\dag,[S,\eta(G_sd_1G_s\cdots d_{p}G_s)]]|x\rangle,
\nonumber
\end{align}
which is diagonal and can be expressed in terms of the product of $G_s$'s. In the continuum, we get (absorbing a factor $N^{-2}$ in the time rescaling):
\beq
 \partial_x^2 \langle x| G_sd_1G_s\cdots d_{p}G_s|x\rangle.
\eeq
\item[$\to$] There are It\^o contractions coming from the Hamiltonian increments at a bulk position. This leads to the sum over $j$ of
\begin{align*}
& -  \frac{1}{2} \langle x| G_sd_1G_s\cdots d_{j-1}G_s)[\dd h_s,[\dd h_s,d_j]](G_sd_{j+1}G_s\cdots d_{p}G_s)|x\rangle\\
=& -\frac{1}{2} \langle x| G_sd_1G_s\cdots d_{j-1}G_s)[S,[S^\dag,d_j]](G_sd_{j+1}G_s\cdots d_{p}G_s)|x\rangle \\
&-\frac{1}{2} \langle x| G_ts_1G_s\cdots d_{j-1}G_s)[S^\dag,[S,d_j]](G_sd_{j+1}G_s\cdots d_{p}G_s)|x\rangle,
\end{align*}
whose scaling limit is  (absorbing a factor $N^{-2}$):
\beq
\langle x| (G_sd_1G_s\cdots d_{j-1}G_s)(\partial^2d_j)(G_sd_{j+1}G_s\cdots d_{p}G_s)|x\rangle.
\eeq
\item[$\to$] There are It\^o contractions between one extremal positions and a bulk one. These terms are the sum over $j$ of 
\begin{align*}
& \langle x| [\dd h_s,(G_sd_1G_s\cdots d_{j-1}G_s)[\dd h_s,d_j](G_sd_{j+1}G_s\cdots d_{p}G_s)|x\rangle\\
= & ~~ \langle x| S\eta(G_sd_1G_s\cdots d_{j-1}G_t)[S^\dag,d_j](G_sd_{j+1}G_s\cdots d_{p}G_s)|x\rangle \\
 &- \langle x| (G_sd_1G_s\cdots d_{j-1}G_s)[S^\dag,d_j]\eta(G_sd_{j+1}G_s\cdots d_{p}G_s)S|x\rangle \\
&+ [\cdots S\leftrightarrow S^\dag \cdots ] .
\end{align*}
In the continuum, this yields (absorbing a factor $N^{-2}$):
\beq
  -2 \partial_x\left({ \langle x| G_sd_1G_s\cdots d_{j-1}G_s)|x\rangle (\partial d_j(x))  \langle x|(G_sd_{j+1}G_s\cdots d_{p}G_s)|x\rangle }\right).
\eeq
\item[$\to$] There are It\^o contractions between two bulk positions. These are the sum over $j<k$ of 
\begin{align*}
&- \langle x|(G_td_1G_s\cdots d_{j-1}G_s)[\dd h_s,d_j] \eta(G_sd_{j+1}G_s\cdots d_{k-1}G_s)[\dd h_s,d_k](G_sd_{k+1}G_s\cdots d_{p}G_s)|x\rangle\\
= &- \langle x|(G_td_1G_s\cdots d_{j-1}G_s)[S,d_j] \eta(G_sd_{j+1}G_s\cdots d_{k-1}G_s)[S^\dag,d_k](G_sd_{k+1}G_s\cdots d_{p}G_s)|x\rangle\\
&- \langle x|(G_td_1G_s\cdots d_{j-1}G_s)[S^\dag,d_j] \eta(G_sd_{j+1}G_s\cdots d_{k-1}G_s)[S,d_k](G_sd_{k+1}G_s\cdots d_{p}G_s)|x\rangle.
\end{align*}
The large $N$ scaling limit is (absorbing a factor $N^{-2}$):
\beq
 2\langle x|(G_sd_1G_S\cdots d_{j-1}G_s)(\partial d_j) \eta(G_sd_{j+1}G_s\cdots d_{k-1}G_t)(\partial d_k)(G_sd_{k+1}G_s\cdots d_{p}G_t)|x\rangle.
\eeq
\end{itemize}

Gathering all terms we get:
\begin{align*}  
\partial_t \hat Q^t_{p+1} [d_1,\cdots,d_{p}] = &~\mathrm{noisy~ terms}
 + \partial^2 \hat Q^t_{p+1}[d_1,\cdots,d_{p}] + \sum_{j=1}^{p} \hat Q^t_{p+1}[d_1,\cdots,\partial^2d_j,\cdots,d_{p}] \\
&- 2\sum_{j=1}^{p} \partial \left({ \hat Q^t_{j}[d_1,\cdots,d_{j-1}]\,  (\partial d_j)\,\hat Q^t_{p+1-j}[d_{j+1},\cdots,d_{p}] }\right)\\
&\hskip -2.0 truecm +2\sum_{j<k} \hat Q^t_{p+1-k+j}[d_1,\cdots,d_{j-1}, 
{\underbrace{(\partial d_j)\hat Q^t_{k-j}[d_{j+1},\cdots,d_{k-1}](\partial d_k)}}
,d_{k+1},\cdots, d_{p}]
\end{align*}

These equations involve a nested non-linear structure, in the second line and in the third line marked with the under-brace. 
In the large $N$ limit we expect factorization of expectation values at leading order in $1/N$, as usual in random matrix theory. This factorization property was implicitly used in \cite{bernard2019open,bernard2021dynamics,hruza2022coherent}. Thus, taking expectation values and using factorization, we get, for $\hat C^t_{p+1}:=\EE[\hat Q^t_{p+1}]$,
\begin{align} \label{eq:dGn-big-mean}
\partial_t \hat C^t_{p+1}[d_1,\cdots,d_{p}] =& ~~ \partial^2 \hat C^t_{p+1}[d_1,\cdots,d_{p}]
 + \sum_{j=1}^{p} \hat C^t_{p+1}[d_1,\cdots,\partial^2d_j,\cdots,d_{p}]  \\
&\hskip -0.5 truecm - 2\sum_{j=1}^{p} \partial\left({ \hat C^t_{j}[d_1,\cdots,d_{j-1}]\, (\partial d_j)\,\hat C^t_{p+1-j}[d_{j+1},\cdots,d_{p}] }\right) \nonumber \\
&\hskip -2.5 truecm +2\sum_{j<k} \hat C^t_{p+1-k+j}[d_1,\cdots,d_{j-1}, 
\underbrace{(\partial d_j)\hat C^t_{k-j}[d_{j+1},\cdots,d_{k-1}](\partial d_k)}
,d_{k+1},\cdots, d_{p}] \nonumber
\end{align}
Recall that $\hat C_{p+1}^t$ are linear in the entries $d_k$ so that the middle nested terms can be written as a (integrated) product of two $\hat C_k^t$'s. These relations should be compared with those satisfied by the free process \eqref{prop:C-eqs-motion}. This proves Proposition \ref{prop:C-motion-discrete-qssep}. 

Let us write them for small $p=0,1,2,\cdots$:\\
-- $p=0$, $\hat C_1^t(x):=\EE[ \langle x| G_s |x\rangle]$, $s=N^2t$, and the equation \eqref{eq:dGn-big-mean} is the heat equation:
\[
\partial_t \hat C_1^t(x) =\partial_x^2 \hat C_1^t(x).
\]
-- $p=1$, $\hat C_2^t[d_1](x)=\EE[ \langle x| G_s d_1 G_s |x\rangle]$, $s=N^2t$, and the equation \eqref{eq:dGn-big-mean} reads:
\[
\partial_t \hat C_2^t[d_1](x) =\partial_x^2 \hat C_2^t[d_1](x) + \hat C_2^t[\partial^2 d_1](x) 
- 2 \partial_x\!\left({ \hat C_1^t(x) (\partial_x d_1)(x)  \hat C_1^t(x) }\right).
\]
-- $p=2$, $\hat C_3^t[d_1,d_2](x)=\EE[ \langle x| G_s d_1 G_s d_2 G_s |x\rangle]$, $s=N^2t$, and the equation \eqref{eq:dGn-big-mean} reads: 
\begin{align*}
\partial_t \hat C_3^t[d_1,d_2](x) =& ~~ \partial_x^2 \hat C_3^t[d_1,d_2](x) + \hat C_3^t[\partial^2 d_1,d_2](x) + \hat C_3^t[d_1,\partial^2 d_2](x) \\
&\hskip -3.0 truecm - 2 \partial_x\!\left({ \hat C_1^t(x) (\partial_x d_1)(x) \hat C_2^t[d_2](x)  }\right) - 2 \partial_x\!\left({ \hat C_2^t[d_1](x) (\partial_x d_2)(x)  \hat C_1^t(x)  }\right) + 2 \hat C_2^t[(\partial d_1)\hat C_1^t (\partial d_2)](x)
\end{align*}


\medskip

\section{Boundary conditions in discrete QSSEP}
\label{app:discrete-bdry}

We here discuss the boundary conditions in the three different cases: "periodic", "closed" or "open".
We first look at the evolution of the mean local densities $n_j:= \EE[ G_{jj}]$. It reads (with $n$ the vector $n:=(n_j)$):
\[
\partial_s  n = A\cdot  n + b ,
\]
where both the matrix $A$ and the linear drift $b$ depend whether we are looking at the periodic, closed or open QSSEP.

(i) "Periodic": The evolution is unitary with Hamiltonian increments associated with the loop algebra $su(N)$. There, $b=0$ and $A$ is the discrete Laplacian on the periodic circle:
\[
A= \begin{pmatrix}  
-2 & 1 & 0 & \cdots & 1 \\ 1 & -2 & 1 & 0 & \cdots  \\  \vdots & ~ & \ddots  & ~ & \vdots \\  \cdots & 0 & 1 & -2 & 1 \\ 1 & \cdots & 0 & 1 & -2 
\end{pmatrix}.
\]
 It is clear that in the continuous limit the boundary conditions are periodic on $[0,1]$.
 \medskip

(ii) "Closed":  The evolution is unitary with Hamiltonian increments involving the simple roots of $su(N)$. There, $b=0$ and $A$ is the discrete Laplacian on the interval $[0,1]$:
\[
A= \begin{pmatrix}  
-1 & 1 & 0 & \cdots & \cdots \\ 1 & -2 & 1 & 0 & \cdots \\  \vdots & ~ & \ddots & ~ & \vdots \\ \cdots & 0 & 1 & -2 & 1 \\ \cdots & \cdots & 0 & 1 & -1 
\end{pmatrix}.
\]
The equations for the densities are $\partial_s n_1=  n_2- n_1$, and  $\partial_s n_j=  n_{j+1}-2 n_j +  n_{j-1}$ for $j=2,\cdots, N-1$, and $\partial_s n_N=  n_{N-1}- n_N$. After rescaling time diffusively $s\leadsto t=N^{-2}s$ and going to the scaling limit $ n_j=\bar n(x=j/n)$ for some smooth function $\bar n(\cdot)$, these equations yield:
\begin{align*}
\partial_t \bar n(0) &= N (\partial_x \bar n)(0) + \cdots ,\\
\partial_t \bar n(x) &= (\partial_x^2 \bar n)(x) + \cdots,\\
\partial_t \bar n(1) &=  -N (\partial_x \bar n)(1) + \cdots .
\end{align*}
The existence of a scaling limit implies the Neumann boundary conditions: $(\partial_x \bar n)(0) =0$ and $(\partial_x \bar n)(1)=0$.
\medskip

(ii) "Open": The discrete evolution is not unitary due to the boundary terms. In the open case, the evolution equation for the density is $\partial_s  n = A\cdot  n + b$ with $b=(\alpha_1,0,\cdots,0,\alpha)^T$, with $n_a=\alpha_1/\nu_1$ and $n_b=\alpha_N/\nu_N$, and
\[
A= \begin{pmatrix}  
-1-\nu_1 & 1 & 0 & \cdots & \cdots \\ 1 & -2 & 1 & 0 & \cdots \\  ~ & ~ & \cdots & ~ & ~ \\ \cdots & 0 & 1 & -2 & 1 \\ \cdots & \cdots & 0 & 1 & -1 -\nu_N
\end{pmatrix}.
\]
The extra drift terms are echos of the boundary Lindbladians. 
After rescaling time diffusively $s\leadsto t=N^2s$, the equations for the densities read:
\begin{align*}
\partial_t n_1 &=  N^2(n_2- n_1)+N^2\nu_1(n_a-n_1), \\
\partial_t n_j &=  N^2(n_{j+1}-2 n_j +  n_{j-1}),\quad \mathrm{ for}\ j=2,\cdots, N-1,\\
\partial_t n_N &=  N^2(n_{N-1}- n_N)+N^2\nu_N(n_b-n_N), 
\end{align*}
In the scaling limit, with $n_j=\bar n(j/N) + (1/N)\tilde n(j/N)+ O(1/N^2)$:
\begin{align*}
\partial_t \bar n(0) &= N^2\nu_1(n_a-\bar n(0)) - N(\nu_1\tilde n(0) -(1-\nu_1) (\partial_x \bar n)(0)) + O(1), \\
\partial_t \bar n(x) &= (\partial_x^2 \bar n)(x) + O(1/N),\\
\partial_t \bar n(1) &=  N^2\nu_N(n_b-\bar n(1)) - N(\nu_N \tilde n(1)+ (\partial_x \bar n)(1)) + O(1),
\end{align*}
If there is a scaling limit, these equations enforces the Dirichlet boundary conditions: $\bar n(0)=n_a$ and $\bar n(1)=n_b$. The sub-leading terms fix $\tilde n$ at the boundary.

More generally, let us look at the SDE satisfied for $G_{ij}$ which, again, are affine, not linear, in the open case:
\[
d G_{ij}= d G_{ij}\big\vert_\mathrm{unitary} + \nu_1\left({ n_a\delta_{1i}\delta_{1j} - \frac{G_{ij}}{2}(\delta_{1i}+\delta_{1j})}\right)\dd s + (\nu_1\leadsto \nu_N) .
\]
After rescaling time diffusively $s\leadsto t=N^{-2}s$, and specifying them to $i=1$ (and writing only the terms associated to the boundary at $x=0$):
\begin{align*}
d G_{11}&= d G_{11}\big\vert_\mathrm{unitary} + N^2\nu_1( n_a - G_{11})dt, \\
d G_{1j}&= d G_{1j}\big\vert_\mathrm{unitary} - N^2\frac{\nu_1}{2} G_{1j}dt
\end{align*}
Here the unitary dynamics is as in the "closed" case.
This implies the boundary conditions $G_{11}\approx n_a$ and $G_{1j}\approx 0$ in the scaling limit. That is: in the scaling, if one of the point is at the boundary but not the other one, then the coherences vanish (no fluctuation).

Notice that the scaling limit of the open QSSEP is $N\to\infty$, fixing the rates $\nu_\alpha\not= 0$ (which are then irrelevant in the limit), while the scaling limit of the closed QSSEP corresponds to a different order of limits: first setting $\nu_\alpha= 0$ and then sending $N\to\infty$. 
\vfill \eject 

\bibliography{QSSEP_Continuum}{}
\bibliographystyle{plain}


\end{document}